\documentclass[final]{siamltex}
\usepackage{amscd, amssymb,latexsym,amsmath,graphics,setspace}
\usepackage{epsfig,epsf,amsfonts}
\usepackage{color,multicol,colordvi,graphicx,multirow}
\usepackage{amsfonts, epsfig}
\renewcommand{\thepage}{\hss\thesection-\arabic{page}\hss}
\normalbaselineskip

\renewcommand{\thepage}{\hss\thesection-\arabic{page}\hss}
\newcommand{\tphi}{\tilde\phi}\newcommand{\tx}{\tilde x}
\newcommand{\tz}{\tilde z}

 \newcommand{\cA}{\mathcal A}
 \newcommand{\cR}{\mathcal R}

\newcommand{\sA}{\mathcal {A}}

\newcommand{\partderiv}[2]{\ensuremath{\frac{\partial #1}{\partial #2}}}

\newcommand{\cnacl}{C_{\textrm{\tiny{NaCl}}}}

\newcommand{\fp}{f_{+}}

\newcommand{\bp}{\mathbf p }

\newcommand{\Hplus}{{\textrm H}^+}

\newcommand{\kw}{{K}_{\textrm{\begin{tiny}w\end{tiny}}}}
 \newcommand{\kh}{K_{\textrm{\begin{tiny}h\end{tiny}}}}

\newcommand{\kgzero}{K_{\textrm{\begin{tiny}G\end{tiny}}}^{\textrm{\begin{tiny}
0\end{tiny}}}}
\newcommand{\Kh}{K_{\textrm{\begin{tiny}H\end{tiny}}}}
\newcommand{\nuw}{\nu_{\textrm{\begin{tiny}w\end{tiny}}}}

\newtheorem{remark}{Remark}

\newcommand{\selast} {s_{\textrm{\begin{tiny} {elast} \end{tiny}}}}
\newcommand{\smix}{s_{\textrm{\begin{tiny} {mix} \end{tiny}}}}
\newcommand {\sion}{s_{\textrm{\begin{tiny} {ion} \end{tiny}}}}
\newcommand{\gmix}{g_{\textrm{\begin{tiny} {mix}
\end{tiny}}}}

\newcommand{\bx}{{\mathbf x}}

\newcommand{\Cna}{C_{\textrm{\begin{tiny}Na\end{tiny}}}}

\newcommand{\Cah}{C_{\textrm{\begin{tiny}AH\end{tiny}}}}
\newcommand{\Ch}{C_{\textrm{\begin{tiny}H\end{tiny}}}}

\newcommand{\Ka}{K_{\textrm{\begin{tiny}A\end{tiny}}}}
\newcommand{\Kg}{K_{\textrm{\begin{tiny}G\end{tiny}}}}

\newcommand{\vw}{v_{\textrm{\begin{tiny}w\end{tiny}}}}
\newcommand{\Ccl}{C_{\textrm{\begin{tiny}Cl\end{tiny}}}}

\newcommand{\Ccchamber}{C_{\textrm{\begin{tiny}C\end{tiny}}}^{\textrm{\begin{
tiny}II\end{tiny}}}}

\newcommand{\Cg}{C_{\textrm{\begin{tiny}G\end{tiny}}}}

\newcommand{\Cnacl}{C_{\textrm{\begin{tiny}NaCl\end{tiny}}}}
\newcommand{\Caminus}{C_{\textrm{\begin{tiny}A-\end{tiny}}}}

\newcommand{\ChI}{C_{\textrm{\begin{tiny}H\end{tiny}}}^{\textrm{\begin{tiny}
I\end{tiny}}}}
\newcommand{\ChII}{C_{\textrm{\begin{tiny}H\end{tiny}}}^{\textrm{\begin{tiny}
II\end{tiny}}}}

\newcommand{\kmar}{k_{\textrm{\begin{tiny}mar\end{tiny}}}}

\newcommand{\Chcominus}{C_{{\textrm{\begin{tiny}H\end{tiny}}}{\textrm{\begin{
tiny}CO\end{tiny}}}_3^{\textrm{\begin{tiny}-\end{tiny}}}}}

\newcommand{\Jhcominus}{J_{{\textrm{\begin{tiny}H\end{tiny}}}_2{\textrm{\begin{
tiny}CO\end{tiny}}}_3^{-}}^{\textrm{\begin{tiny}II\end{tiny}}}}

\newcommand{\by}{{\mathbf y}}\newcommand{\eps}{{\epsilon}}

\newcommand{\bt}{{\bar t}}

\newcommand{\gelast}{{g_{\textrm{\begin{tiny}elast\end{tiny}}}}}

\newcommand{\Gelast}{{G_{\textrm{\begin{tiny}elast\end{tiny}}}}}

\newcommand{\Gmix}{{G_{\textrm{\begin{tiny}mix\end{tiny}}}}}
\newcommand{\gelastuni}{{g^{\textrm{\begin{tiny}uni\end{tiny}}}_{\textrm{\begin
{tiny}elast\end{tiny}}}}}

\begin{document}

\author{ 
Lingxing Yao\\
Department of Mathematics, Applied Mathematics and Statistics\\ Case Western Reserve University\\
Cleveland, OH 44106
\\
{ \scriptsize{and}}\\
M.~Carme Calderer\and Yoichiro Mori\\School of Mathematics\\ 
{\scriptsize{and}}\\
 Ronald A.~Siegel
 \\
 Departments of Pharmaceutics and Biomedical Engineering\\
 University of Minnesota, Minneapolis, MN 55455
 }

 \title{Rhythmomimetic drug delivery: modeling, analysis and numerical simulation}

\date{\today}
\maketitle
\bibliographystyle{siam}
\thanks{M. Carme Calderer acknowledges the  support of the National Science Foundation through the grants DMS-0909165 and DMS-GOALI-1009181.
Ronald A. Siegel acknowledges support from the National Institutes of Health through grant HD040366.}


\begin{abstract}We develop, analyse and numerically simulate  a model of a prototype, glucose-driven, rhythmic drug delivery device, 
aimed at hormone therapies,  and based on chemomechanical interaction in a polyelectrolyte gel membrane. The pH-driven  interactions   trigger volume phase transitions 
between the swollen and collapsed states of the gel.	For a robust set of material parameters, we find a class of solutions of the governing system that
oscillate between such states, and cause the membrane to rhythmically swell, allowing for transport of the drug, fuel and reaction products across it, and collapse,
hampering all transport across it. The frequency of the oscillations can be adjusted so that it matches the natural frequency of the hormone to be released.  
The work is linked to extensive laboratory experimental studies of the device built by Siegel's team. The thinness of the membrane and its clamped boundary, 
together with the homogeneously held conditions in the experimental apparatus,  justify  neglecting spatial dependence on the fields of the problem.   
Upon identifying the forces and energy relevant to the system, and taking into account its dissipative properties, we apply  Rayleigh's 
variational principle to obtain the governing equations.  The material assumptions guarantee the monotonicity of the system and lead to the existence of a 
three dimensional limit cycle. By scaling and asymptotic  analysis, this limit cycle is  found to be related to a two-dimensional one that encodes the volume phase 
transitions of the model.  The identification of the relevant parameter set of the model is aided by a Hopf bifurcation study of steady state solutions. 
\end{abstract}

\section {Introduction}\label{intro}
Research on drug delivery systems is focused on presenting drug at the right place in the body at the right time \cite{DrugDelivery}. 
To this end, we have introduced a prototype chemomechanical oscillator that releases GnRH in rhythmic pulses, fueled by exposure to a constant level of glucose \cite{dhanarajan2002autonomous,dhanarajan2006,Misra}. Experience with chemical and biochemical oscillators \cite{GrayScott}, \cite{EpsPoj} and \cite{Goldbeter}, and with
electrical and mechanical relaxation oscillators \cite{Adronov}, shows that rhythmic behavior can be driven by a constant rate stimulus, provided proper delay, memory and feedback elements 
are employed in device dynamics.
 The device consists of two fluid compartments, an external cell (I) mimicking the physiological environment, and a closed chamber (II), separated from (I) by a hydrogel membrane.  Cell I, which is held at constant pH and ionic strength, provides a constant supply of glucose to cell II, and also serves as clearance station for reaction products. Cell II contains the drug to be delivered to the body, an enzyme that catalyzes conversion of glucose into hydrogen ions, and a piece of marble to remove excess hydrogen ions that would otherwise overwhelm the system. When the membrane is swollen, glucose flux into Cell II is high, leading to rapid production of hydrogen ions.  However, these ions are not immediately released to Cell I but react, instead, with the negatively charged carboxyl groups of the membrane, which collapses when a critical pH is reached in Cell II, substantially attenuating glucose transport and production of ions.  Subsequent diffusion of membrane attached H-ions increases again the concentration of negative carboxyl groups in the membrane, causing the gel to reswell, and so, the process is poised to repeat itself.  Since drug release can only occur when the membrane is swollen,
 it occurs in rhythmic pulses that are coherent with the pH oscillations in Cell II and swelling oscillations of the membrane. 
 
 While rhythmic hormone release across the membrane is the ultimate goal of this research, a main  purpose of this article is the study of the pH oscillations in Cell II. These oscillations correlate with those of $\Hplus$ concentration in the membrane, which determine its swelling state.
 
 A polyelectrolyte gel is a mixture of polymer and fluid,   the latter containing several species of  ions  and the polymer including  electrically charged (negative) side groups. We model the gel  as an incompressible, saturated  mixture of polymer and fluid.
The hydrogel (polyelctrolyte gel having water as the fluid solvent) of the experimental device is a relatively thin membrane that is laterally restrained by clamping justifying its treatment as a one-dimensional system.
Moreover, 
the experimental apparatus is kept well-stirred at all times,  allowing for further reduction  to a time-dependent,  spatially homogeneous system. 
The system consists of three coupled mechanical-chemical ordinary differential equations for  the time evolution of the  membrane thickness $L=L(t)$,
the hydrogen concentrations $x (C_H^M)$ and $z (C_H^{II})$ inside  the membrane and in the chamber, 
respectively.   The polymer volume fraction $\phi(t)$ is related to $L(t)$ by the equation   of balance of  mass of the polymer, $L=\frac{\phi_0}{\phi}$,
 where $\phi_0$ denotes the corresponding volume fraction  in the reference state of the membrane. 
The governing system also takes into account the algebraic constraint of electroneutrality.
Analysis and simulation of gel swelling in higher dimensional geometries, in the purely mechanical case,  have been carried out in \cite{chabaud-calderer14} and \cite{Suo}.

The swelling force in the hydrogel consists of mixing, elastic and ionic terms \cite{English, Flory}. The first two terms are derived from the Flory-Huggins 
energy of mixing and the neo-Hookean form of the elastic stored  energy, respectively. 
The combination of Lagrangian elasticity and Eulerian mixing energies gives rise to the well known Flory-Rehner theory.  The ionic force follows  Van't Hoff's ideal law.
The latter force depends on the degree of ionization due to acid-base equilibria between pendant ionizable groups in the hydrogel and free hydrogen ion,
as described by a Langmuir isotherm, and the resulting Donnan partitioning of counterions and co-ions into the hydrogel. 
Taking all these factors into account, we call the model for swelling stress the Flory-Rehner-Donnan-Langmuir (FRDL) model.

A scaling analysis reveals that, of the many physical parameters of the model, five dimensionless parameter, $\mathcal A_i$, encode most  mechanical and chemical
properties of the system. These  
together with control parameters such as salt concentration $\cnacl$, reaction rate constant of the marble $\kmar$, degree of ionization of the polymeric side groups $\sigma_0$, reference polymer volume fraction $\phi_0$  and degree of polymer cross-linking $\rho_0$ are sufficient to
fully describe the evolution properties of the system. 

A study of the stationary states of the system shows the hysteretic behavior for appropriate parameter ranges that is consistent with laboratory 
experiments \cite{BakerSiegel1996, FirestoneSiegel1988, BhallaSiegel2014}. 
Moreover, we find that, for every set of parameters  within the relevant  range, there exists a unique steady state that is hyperbolic. 
We study the Hopf bifurcation of the steady state and numerically 
identify the parameter regions within which oscillations occur. We found these regions to be in good agreement with experiment.

Dimensionless   parameter groups determine the relative time scales of the system, with the fast time corresponding to the swelling ratio of the membrane. 
 In particular, this implies  that solutions of the governing system remain arbitrarily close to those of a suitable two dimensional system, for most of the time. 
 Although the original system has  positive $\textrm C^1$-solutions with $\phi$ bounded away from
 0  (and $\phi$ bounded away from $1$, as well), 
the two-dimensional restriction involves multivalued functions, leading to existence of weak solutions, with discontinuous $\phi$ and, consequently,  
the  swelling ratio  being discontinuous as well. 
The latter corresponds to a volume phase transition taking place that drives the permeability changes of the membrane. 
The Poincar{\'e}-Bendixon theory for plane systems applies to the current  problem 
from which existence of a two-dimensional limit cycle follows.  It turns out that the limit cycle is also the omega-limit set of positive semi-orbits of the full system. 

In order to show existence of a limit-cycle of the three dimensional system, we must  appeal to the theory of monotone dynamical systems. 
Specifically, we show that our three dimensional system is competitive
with respect to a properly defined alternate cone, that results from the intersection of half-spaces. We also find that the Jacobian matrix of the 
system satisfies the properties of {\it sign-stability} and 
{\it sign-symmetry}, from which  existence of a three-dimensional limit cycle follows.  
It turns out that our system is qualitatively analogous to that modeling the dynamics of the HIV  in the regime where the concentrations of infected and uninfected T-cells and the virus follow a 
periodic holding pattern,  away from the fully infected state represented by a hyperbolic steady state \cite{virus}.

The outline of this paper is as follows.  In section \ref{hormone-delivery}, we  describe the prototype oscillator and outline its basis for operation.  In section \ref{Parameter-Model}, we 
study the mechanical and chemical properties of the system. The assumptions of the model and the subsequent  formulation of a precursory system  of ordinary differential equations are presented 
in section \ref{chemo-mechanical model}. In section \ref{scaling},  we describe the parameters of the system and  explore its scaling
properties, from which the  system of ordinary differential equations that  model the dynamics of the oscillator follows. The Hopf bifurcation analysis of the unique hyperbolic equilibrium point is presented in 
section \ref{numerical-simulations}. In section \ref{inertial-manifold}, we construct an inertial manifold of the three-dimensional system and  
analyze the  corresponding two-dimensional flow
in the manifold, demonstrating  existence of an asymptotically stable limit cycle. 
   In section \ref{3D-limit-cycle},  we study the monotonicity properties of the system and prove the existence of a three-dimensional limit cycle.  A multiscale, asymptotic analysis  is developed in section \ref{multiscale} that yields decay estimates of solutions as they approach the two-dimensional manifold, where hysteresis properties manifest themselves, including the can{\`a}rd  feature of the system. 
   We conclude with a discussion of the benefits and deficiencies of the present model.

\section{ Rhythmic Hormone Delivery: A Simple Experimental System}\label{hormone-delivery}
A simplified schematic representation of the experimental oscillator is depicted in Figure 1 \cite{dhanarajan2002autonomous, Misra}. The apparatus consists of two well stirred fluid compartments, or cells, separated by a sweallable hydrogel membrane.  Cell I, which is meant to mimic the external physiological fluid environment, contains glucose in  saline solution, with fixed glucose concentration maintained by introduction of fresh medium and removal of reaction products by flow, and fixed pH 7.0 enforced by a pH-stat servo (autotitrating burette). Cell II, which simulates the interior of the rhythmic delivery device, contains the hormone to be released (e.g. GnRH) and the enzymes glucose oxidase, catalase, and gluconolactonase, which catalyze conversion of glucose to gluconic acid.  The latter rapidly dissociates into hydrogen ion (${\textrm H}^+$) and gluconate ion:
$$ \textrm{Glucose} +{\textrm O}_2 \stackrel{\textrm{\begin{tiny}\emph{enzymes}\end{tiny}}}{\longrightarrow} {\textrm H}^+ +{\textrm {Gluconate}}^- + \frac{1}{2}{\textrm O}_2    \textrm{   (I)}$$
 \begin{figure}
 \label{Figure1}
\centerline{
{\includegraphics[width=90mm]{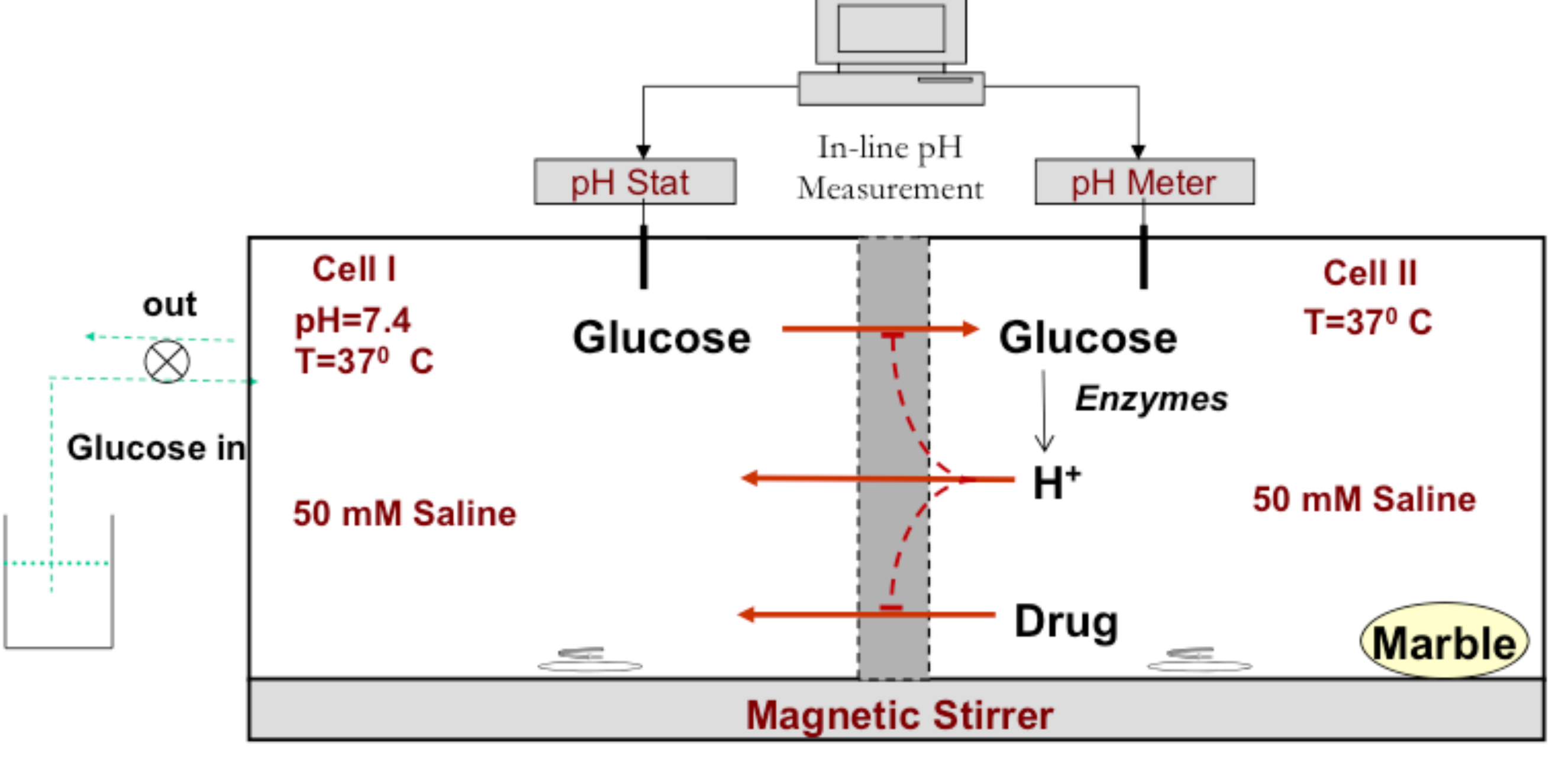}}}
\caption{\small{Diagram of the Experimental Oscillator  \cite{dhanarajan2002autonomous, Misra} }}
\end{figure}
\noindent Cell II also contains physiologic saline, which is exchanged with Cell I through the membrane, and a piece of solid marble.  Marble is solid calcium carbonate,
 ${\textrm{ CaCO}}_3(\textrm{s}), $ which reacts with $\textrm{H}^+$ according to $$2{\textrm H}^+ + {\textrm{ CaCO}}_3(\textrm{s})    \longrightarrow        
{\textrm{ Ca}}^{2+} + {\textrm{ CO}}_2\uparrow  \,\, +{\textrm{
H}}_2\textrm{O} \textrm{ (II)}$$
The hydrogel membrane is clamped between Cells I and II.  The degree of swelling of the  membrane and permeabilities to
glucose and GnRH, depend on the internal concentration of ${\textrm H}^+$  ions, through the reaction \cite{English, grimshaw1990kinetics, ricka-tanaka}.
 $$\vdash{\textrm{COO}}^-  + {\textrm{H}}^+ \rightleftharpoons {\textrm{ }} \vdash{\textrm{COOH}} \textrm{  (III)} $$
\noindent At low ${\textrm{H}}^+$ concentration, the membrane is charged, swollen, and highly permeable to both glucose and GnRH,
but permeability to these compounds is substantially attenuated at higher $\Hplus$ concentrations where the membrane has less charge and is relatively collapsed.
The placement of the membrane between a  source of fuel, glucose, and its converting enzymes creates  a dynamic  environment with competing effects. On the one hand, the enzyme reaction produces $ \Hplus, $ which affects the permeability of the membrane; it is  removed from the system by reacting with marble, and also by diffusing out to the environment. This creates a negative feedback mechanism between the enzyme reaction and the membrane permeability to glucose. Under proper conditions, this arrangement can lead to oscillations.

%

When the membrane is ionized and swollen, glucose permeates from Cell I to Cell II and is converted to  $\Hplus,$ which diffuses back into the membrane,
 binds and neutralizes the $\vdash\textrm{COO}^-$ groups, and causes the hydrogel membrane to collapse.  The membrane is now impermeable to glucose, and enzymatic production of
 $\Hplus$ is attenuated.  Eventually the $\Hplus$ ions bound to the membrane diffuse into Cell I, where they are neutralized by the pH-stat  and removed in the waste stream.  The membrane then reionizes and reswells, and the system is primed to repeat the previous sequence of events.

In order to achieve sustained oscillations, a steady state in which flux of glucose, enzyme reaction rate, and flux of $\Hplus$ are balanced and equal at all times,
 must be avoided.  As it will be seen, bistability, or hysteresis of membrane swelling response to $\Hplus,$ provides a means for destabilizing such a steady state.

The reader is referred to experimental details and results in previous publications
\cite{dhanarajan2002autonomous, dhanarajan2006, Misra}.

\section{ A Lumped Parameter Model}\label{Parameter-Model}

A full mathematical description of the experimental system just described would require a detailed account of 1) diffusional and convective fluxes of solvent and solutes, 2) the spatial three-dimensional mechanics of the hydrogel membrane which, though constrained by clamps, exhibits swelling and shrinking which is both time and position dependent, 3) the kinetics of enzymatic conversion of glucose to $\Hplus$, and 4) the kinetics of binding and dissociation of $\Hplus$ with $\vdash\textrm{COO}^-$ groups.  An accurate, verifiable model of this sort, which would require partial differential equations to describe intramembrane processes, does not yet exist.  Here we simplify the problem by assuming that the membrane is a lumped, uniform element. All mechanical and chemical variables are homogenized to single values representing the whole membrane.  We recognize that some potentially important consistencies are lost in this approximation.  First, there will always be a difference in pH between Cells I and II, which will lead to intramembrane gradients in the chemical and mechanical variables.  Second, self consistent boundary conditions, which would follow naturally from a  PDE model, must be replaced by somewhat \emph{ad hoc} assumptions.

As a second simplification, we assume that the enzyme reactions, and the process of distribution of the dominant background electrolyte, NaCl, between Cells I and II and the hydrogel membrane, are very fast compared to the other dynamical processes.  These assumptions can be justified, respectively, by the excess of enzyme used in the experiments, and the fact that the capacity of the membrane for acidic protons ($\Hplus,$ or $\vdash$COOH) relative to Cells I and II, far exceeds its relative capacities for $\textrm{Na}^+$ and $\textrm{Cl}^-$.  Many experimental studies with polyacidic hydrogels have confirmed that $\Hplus$ dynamics and poroelastic relaxations are much slower than those of NaCl \cite{FirestoneSiegel1988, BakerSiegel1996, BhallaSiegel2014}.
\subsection{Swelling of Hydrogels}
The membranes considered in this work are crosslinked networks of polymer
chains, or hydrogels, which absorb substantial amounts of water.  Depending on
water content, or degree of swelling, the hydrogel will be more or less
permeable to solutes such as glucose and $\Hplus$.  Hydrogels have a long history
of application in drug delivery and medicine due to their mechanical and
chemical compatibility with biological tissues and their ability to store and
release drugs in response to environmental cues \cite{Peppas}.

In the present system, we utilize the hydrogen ions $\Hplus$ that are
enzymatically generated from glucose to control hydrogel swelling and hence release of hormone.   In polyacid hydrogels, swelling is controlled by degree of ionization, which results from dissociation of acidic side groups that are attached to the polymer chains. When NaCl is present in the aqueous fraction of the hydrogel, the ionization equilibrium is represented by
$$\vdash\textrm{COOH} + {\textrm{ Na}}^+ +  {\textrm{Cl}}^-   \rightleftharpoons\,
\vdash{\textrm{COO}}^-{\textrm{ Na}}^+ +  {\textrm{H}}^+  + {\textrm{Cl}}^-  $$ 
Swelling of a polyacidic hydrogel results from three thermodynamic
driving forces \cite{English, Flory,Katchalsky,RickaTanaka}.  First there is the tendency of solvent (water) to enter the
hydrogel and mix with polymer in order to increase translational entropy. 
The mixing force also depends on the relative molecular affinity or aversion of the
polymer for water compared to itself due to short range van der Waals, hydrogen
bonding, and hydrophobic interactions.   Second, there is an elastic force, which is a response to the change in conformational entropy of the polymer chains that occurs during swelling or shrinking.  The third force is due to ionization of the acidic pendant groups, which leads to an excess of mobile counterions and salt inside the
hydrogel compared to the external medium, promoting osmotic water flow into the
hydrogel.  The ion osmotic force acts over a much longer range than the direct polymer/water interaction.

In the present work, it is assumed that the hydrogel is uniform in composition
when it is prepared.  The hydrogel at preparation is taken as the reference
state for subsequent thermodynamic model calculations.  At preparation, the
volume fraction of polymer in the hydrogel is denoted by $\phi_0$.  Crosslinking leads to an initial density of
elastically active chains, $\rho_0$.  The initial density of ionizable acid groups, fixed to
the polymer chains, is denoted by $\sigma_0$.  Both  densities  have units mol/L, and are
referred to the total volume of the hydrogel (polymer+water) at preparation.
\subsection{Mechanics of Swelling}
The swelling state of a hydrogel is characterized by the principal stretches or elongation ratios $ (\alpha_x,\alpha_y,\alpha_z)$, and the volume fraction $\phi$ of the polymer. 
We  let $0<\phi_0<1$ and $L_0>0$ denote  the volume fraction of  polymer and the thickness of the membrane, respectively,  in the reference state.   The   equation of conservation of mass of polymer in the gel  is 
\begin{equation}\phi=\phi_0(\alpha_x\alpha_y\alpha_z)^{-1}. \label{mass-balance-equation} \end{equation} 
The swelling ratio of the membrane relative to its undeformed reference state is $\phi_0/\phi=\alpha_x\alpha_y\alpha_z$, that is,  the Jacobian of the deformation map from the reference to the deformed membrane.  We assume that the elongation ratios and volume fraction vary with time. 
Since in the present system the hydrogel is a relatively thin membrane that is laterally restrained by clamping, we  assume that the main swelling effect occurs along the thickness direction.  Hence, we consider uniaxial deformations 
\begin{equation}
\alpha:=\alpha_x, \alpha_y=\alpha_z=1. \label{uniaxial}
\end{equation}
\begin{remark}  A rigorous  justification of the uniaxiality assumption on the membrane deformation may follow from  a dimensional reduction analysis. However, a main drawback 
of the one-dimensional treatment is that it neglects  undulating  surface instabilities observed in gel free surfaces,  which may ultimately hamper  sustained oscillatory behaviour. 
\end{remark}
Denoting by  $L=L(t)$  the thickness of the membrane at time $t\geq 0$, 
equation (\ref{mass-balance-equation}) reduces to 
\begin{equation}
L\phi=L_0\phi_0. \label{mass-balance-1d}
\end{equation}
In the model presented below, the three swelling forces, mixing, elastic and
ionic,  yield  the  uniaxial  {\it swelling stress}, 	
\begin{equation}	
s=s_{\textrm{\begin{tiny} {mix} \end{tiny}}} +s_{\textrm{\begin{tiny} {\textrm elast}\end{tiny}}}+  s_{\textrm{\begin{tiny} {\textrm ion} \end{tiny}}}. \label{ras1}
\end{equation} 
which reflects the excess free energy density of the  hydrogel  relative to equilibrium, at a
given state of swelling in a prescribed aqueous medium.  At equilibrium, $s=0$. 
Let $\gmix $ and $ \gelast $  denote the Flory-Huggins mixing and  elastic  energy densities with respect to deformed volume, respectively. The corresponding  total energy quantities, per unit cross-sectional area, are
\begin{equation}
\Gmix= L\gmix(\phi), \quad 
\Gelast=L\gelast(\phi), \label{ElasticAndMixingEnergies}
\end{equation} 
A standard variational argument gives the dimensionless stress components as 
\begin{equation}
\smix=-\frac{d\Gmix}{dL}, \quad \selast=-\frac{d\Gelast}{dL}.
\end{equation}
 Consequently, 
							\begin{equation}s_i=-(g_i-\phi\frac{\partial g_i}{\partial \phi}), \,\, i=\textrm{\scriptsize{mix}}, \textrm{\scriptsize{elast}}. \label{ras2}
\end{equation}
The mixing free energy density of solvent (water) with the hydrogel chains is
modeled according to the Flory-Huggins expression
\begin{equation} \gmix = \frac{RT}{\vw}[(1- \phi)\ln(1- \phi) + \chi\phi(1-
\phi)], \label{ras3}   \end{equation}
where $R $ is the gas constant, $T$ is absolute (Kelvin) temperature, and $\vw$ is
the molar volume of water.  The first term in the square brackets accounts
for the translational entropy change of solvent molecules as they move from the
external environment into the hydrogel, temporarily assuming that the bath is
pure solvent.  The second term accounts for short range contact interactions
between polymer and solvent, summarized by the dimensionless parameter $\chi$,
which is the molar free energy required to form solvent/polymer contacts,
normalized by $RT$.  Introducing (\ref {ras3}) into (\ref{ras2}) yields
		\begin{equation}
\smix=-\frac{RT}{\vw}[\ln(1-\phi)+\phi+\chi\phi^2].\label{ras4}
\end{equation}
In the simplest form of Flory-Huggins theory $\chi$  is constant, but in many hydrogel systems, especially those that undergo sharp swelling/shrinking transitions, a pseudo-virial series in $\phi$ is used.  In the present work the series is truncated at the linear term, taking the form \cite{ErmanFlory}, \cite{Hirotsu},							
	\begin{equation}
\chi=\chi_1+ \chi_2\phi. \label{ras5}
\end{equation}
\noindent With  $\chi_2>0$, polymer and solvent become increasingly incompatible as polymer
concentration increases, with hydrogel shrinking, and mixing pressure decreasing.

We model the elastic energy density according to the Neo-Hookean expression
\begin{equation}
\gelast=\frac{RT}{2}\rho_0(\frac{\phi}{\phi_0}) (\alpha_x^2+\alpha_y^2+\alpha_z^2-3-\ln\alpha_x\alpha_y\alpha_z). \label{ras6}
\end{equation}
The logarithmic term accounts for change in entropy of localization of
crosslinks in the hydrogel due to swelling. For the uniaxial swelling under consideration,  
\begin{equation}
\gelastuni= RT\rho_0(\frac{\phi}{\phi_0})[(\frac{\phi_0}{\phi})^2-\ln\frac{\phi_0}{\phi}-1]	\label{ras7}
\end{equation}
 which upon introduction into (\ref{ras2}) yields
\begin{equation}
\selast=-RT\rho_0(\frac{\phi_0}{\phi}-\frac{\phi}{2\phi_0}).\label{ras8}
\end{equation}
Summing (\ref{ras4}) and (\ref{ras8}) we obtain the Flory-Rehner expression for the nonionic component of the swelling stress (see also (\ref{equation-mechanics}), below.  
We now derive the evolution equation of the membrane according to the minimum rate of dissipation principle.
\begin{proposition}
Suppose that  $\phi=\phi(t)$ and $L=L(t), \, t\geq 0$ satisfy equation (\ref{mass-balance-1d}).
  Let us define the rate of dissipation function and the total energy as 
\begin{equation}
 W=\frac{K}{2} (\frac{dL}{dt})^2\, {\textrm{\small{and}}} \, \, \, E=\Gmix+\Gelast,
\end{equation}
respectively,  with $K>0$ denoting a friction coefficient.  Then the function  $v=\frac{dL}{dt}$ that  minimizes the Rayleighian functional $R= \dot E +W,$ among spatially homogeneous fields, satisfies the equation
\begin{equation}
 \smix+ \selast -K\frac{dL}{dt}=0, \label{motion-equation}
\end{equation}
with $\smix$ and $\selast$ as in (\ref{ras4}) and (\ref{ras8}), respectively, and $\Gelast$ and $\Gmix$ as in  (\ref{ElasticAndMixingEnergies}). 
 \end{proposition}
The proof results from the  simple calculation
\begin{equation}
\frac{dE}{dt}=\frac{d}{dt}\big(\gmix+\gelast\big), \quad
= -(\smix+\selast)v, \quad v=\frac{dL}{dt}. \label{crit-pointR} 
\end{equation}
So,  the critical points $v=v(t)$ of $R$ satisfy equation (\ref{motion-equation}),  for all $t\geq 0$.
Defining the {\it permeability coefficient} as
$
\kw= \frac{RT}{K\nu_w},$
 and using equations (\ref{ras4}) and (\ref{ras8}),  equation  (\ref{motion-equation}) becomes
\begin{equation}
\frac{dL}{dt}=   -K_w[ \ln(1-\phi)+\phi+\chi\phi^2 + \nu_w\rho_0(\frac{\phi_0}{\phi}-\frac{\phi}{2\phi_0})].\label{equation-mechanics}
\end{equation}
Equation (\ref{motion-equation}) should be expanded to include the force $\sion$   as in (\ref{ras1}).
The calculation of such a contribution is given next. 
\subsection{Chemical reactions}
While the mixing and elastic terms represent important contributions to swelling
pressure, the ionic term responds to $\Hplus$
 concentration inside the membrane, and this is what
forms the basis for the oscillator's dynamic behavior.  As described above, 
ionization of the hydrogel occurs by dissociation of pendant carboxylic acid
groups.  The fraction of these groups that are ionized, $f$, is modeled
according to a Langmuir isotherm relation,
			\begin{equation} f=\frac{\Ka}{\Ka+\Ch}.
\label{ras9}\end{equation}	 
where $ \Ch $ is the concentration of  $ \Hplus $
 in the aqueous portion of the hydrogel and $K_A$
is the dissociation constant of the carboxylic acid.  The concentration of
ionized groups, referenced to the aqueous portion of the hydrogel, is then given by
\begin{equation}
\Caminus=f\sigma_0(\phi/\phi_0), \label{ras10}
\end{equation}	
where  $\sigma_0$ denotes the reference density of ionized groups fixed to the polymer chains. 
Letting $\Cah $ denote the concentration of $H^+$ linked to carboxyl groups, the conservation of intra-membrane hydrogen ions  is given by the equation
\begin{equation}
L(\Ch+\Cah)= L_0\sigma_0. \label{polymer-charge-concentration}
\end{equation}
Combining equations (\ref{ras10}) and (\ref{polymer-charge-concentration})
yields
\begin{equation}
\Cah= \frac{\phi}{\phi_0}(1-f)\sigma_0. \label{d4}
\end{equation}  
Ionization of hydrogel side-chains, plus the requirement for
quasi-electroneutrality over any distance exceeding a few Debye lengths ($\sim 5$ nm
in physiological systems) leads to a distribution of mobile ions in the hydrogel
that differs from that in the external bath. Ignoring very small contributions
from $ \Hplus $ and $ \textrm{OH}^{-} $ ions, a quasineutrality condition is
		\begin{equation}\Cna-\Ccl-\Caminus=0. \label{ras11} \end{equation}
\noindent where $\Cna$ and $ \Ccl$ are the concentrations, respectively, of $\textrm{Na}^{+}$ and $\textrm{Cl}^{-}$ in the aqueous portion of the hydrogel.

For simplicity, we assume that Cells I and II contain fully dissociated NaCl, with
ion concentrations  $\Cna'=\Ccl'=\Cnacl'$.
Ionic swelling stress in the membrane is modeled using van't Hoff's ideal law:
						
\begin{equation}\sion=RT(\Cna+\Ccl-2\Cnacl').     	\label{ras12}  \end{equation}

\noindent Derivation of this term from a free energy density expression is possible but
not carried out here since the basis for van't Hoff's law is well
understood.
	
It can be shown that diffusion of salt inside the hydrogel is very
fast compared to diffusion of $ \textrm{H}^+$ and mechanical relaxation of swelling
pressure.  We may therefore assume that at any point in the hydrogel a Donnan
quasi-equilibrium, where $ \Cna=\lambda\Cna' $  and  $ \Ccl=\lambda^{-1}\Ccl '$,
with $\lambda$ the Donnan ratio, determined by combining
(\ref{ras9})-(\ref{ras11}), giving
		\begin{equation}(1-\phi)(\lambda-\frac{1}{\lambda})\Cnacl'
-\sigma_0(\phi/\phi_0)f=0.	
\label{ras13} \end{equation}
 With $\lambda$ in hand, (\ref{ras12}) becomes
	\begin{equation}\sion=RT(\lambda+\frac{1}{\lambda}-2)\Cnacl'	.
\label{ras14} \end{equation}
Introducing (\ref{ras4}, \ref{ras5}, \ref{ras8}, \ref{ras14}) into (\ref{ras1}) gives
\begin{equation} \frac{\vw
s}{RT}=-[\ln(1-\phi)+\phi+(\chi_1+\chi_2\phi)\phi^2]-\vw\rho_0(\frac{\phi_0}{\phi}-\frac{\phi}{
2\phi_0})+(\lambda+\frac{1}{\lambda}-2)\vw\Cnacl', 	\label{ras15}
\end{equation}
 which, combined with (\ref{ras13}), determines the total swelling
stress. Rewriting equation (\ref{equation-mechanics}) taking the total stress into account leads to the force balance equation stated in (\ref{d1}).  

Because (\ref{ras15}) and (\ref{ras13}) include elements of Flory-Rehner theory,
and Donnan and Langmuir quasi-equilibria, we call these two equations the FRDL
model for uniaxial swelling stress.  The behavior of this model depends on the hydrogel parameters $\phi_0$, $\rho_0$, $\sigma_0$, $\chi_1$, $\chi_2$  and $p\Ka$, and the external salt concentration $\Cnacl'$, which is assumed to be constant.  These, together with relevant dimensionless parameter groups, determine the swelling and permeability properties of the hydrogel membrane. 
\subsection{Swelling equilibria}\label{swelling-equilibria}
Before pursuing dynamics, it is useful to examine homogeneous uniaxial swelling equilibria for a hydrogel in a bath of fixed $pH=-\log\Ch'$.  In this case $ s=0
$ in (\ref{ras15}) and $\Ch=\lambda\Ch'=\lambda 10^{-pH}$.  Setting $p\Ka:=-\log
\Ka$, equation (\ref{ras13}) becomes
\begin{eqnarray}
&&10^{-(pH-p\Ka)}(1-\phi)\Cnacl'\lambda^3+
(1-\phi)\Cnacl'\lambda^2-[10^{-(pH-p\Ka)}(1-\phi)\Cnacl'\nonumber\\
&&\quad\quad\quad+\sigma_0(\frac{\phi}{\phi_0})]\lambda-(1-\phi)\Cnacl'=0.
\label{ras18} \end{eqnarray}
We discuss the solvability of equation (\ref{ras18}) with respect to $\lambda$.
Since the physical parameters contributing to the polynomial coefficients are positive and $\phi\leq1$, there is only one sign change in the descending polynomial coefficients, and Descartes' rule of signs mandates a single positive real root.  (Negative or complex values of $\lambda$ are physically meaningless as they would predict negative or complex ion concentrations inside the hydrogel.) 
Moreover, substituting $\lambda=\nu+1$  in (\ref{ras18}) and rearranging terms, we obtain a cubic polynomial in $\nu$
that also has only one sign change, hence $\nu\geq 0$ and $\lambda\geq1 $.  This makes physical sense, since for
negatively charged hydrogel the internal cation concentration must exceed its
external counterpart.  

The first and third terms of  (\ref{ras18}) can vary over several orders of magnitude with pH, indicating that pH, which affects ionization, will have a strong influence on salt ion partitioning and hence ion
swelling pressure.  When $ pH<<p\Ka$, the polymer is uncharged, $f\to 0$ and $\lambda \to 1$, and swelling takes on a minimum value determined by the mixing and elastic stresses. When $ pH>> p\Ka$, all acid groups are ionized with  $f\to 1$, and  swelling is maximal with $\lambda $ substantially greater than $1$.

Figure \ref{Ronfig1}a displays the effect of fraction ionized, $f,$ on uniaxial swelling ratio $\alpha$=
$\chi_1 =0.48,$ $\chi_2 =0.60.$  This figure was generated using equation (\ref{ras15})
with $ s=0$ and a version of (\ref{ras13}).
\begin{figure}
\begin{minipage}[t]{\textwidth}
{
\includegraphics[scale=.50]{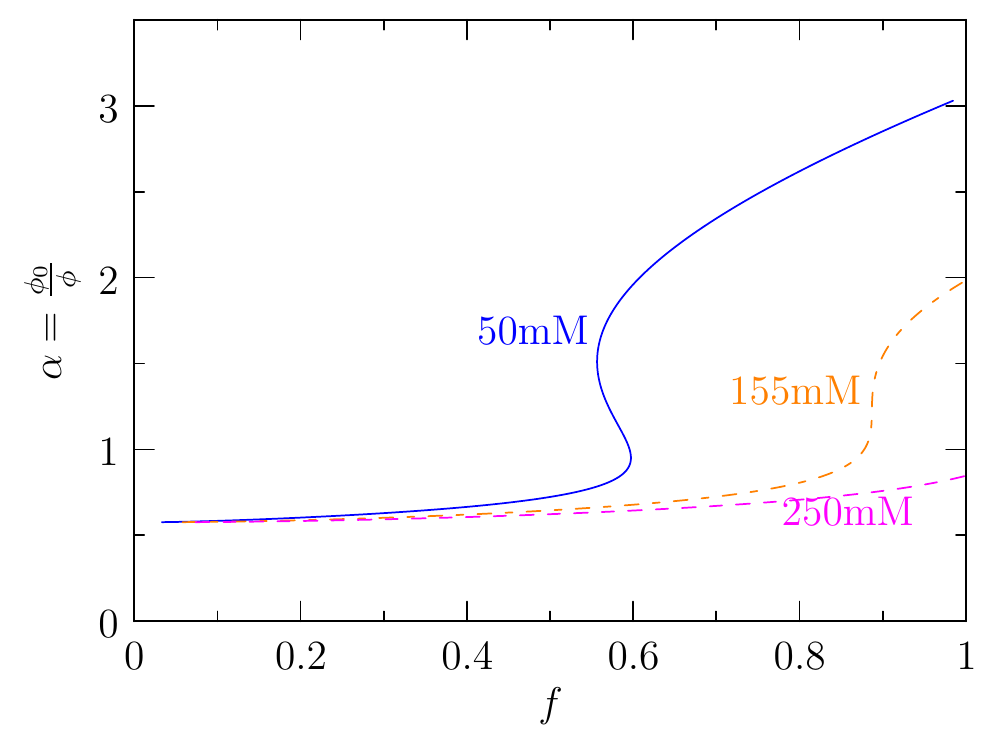}
\hspace{.4cm}
\includegraphics[scale=.50]{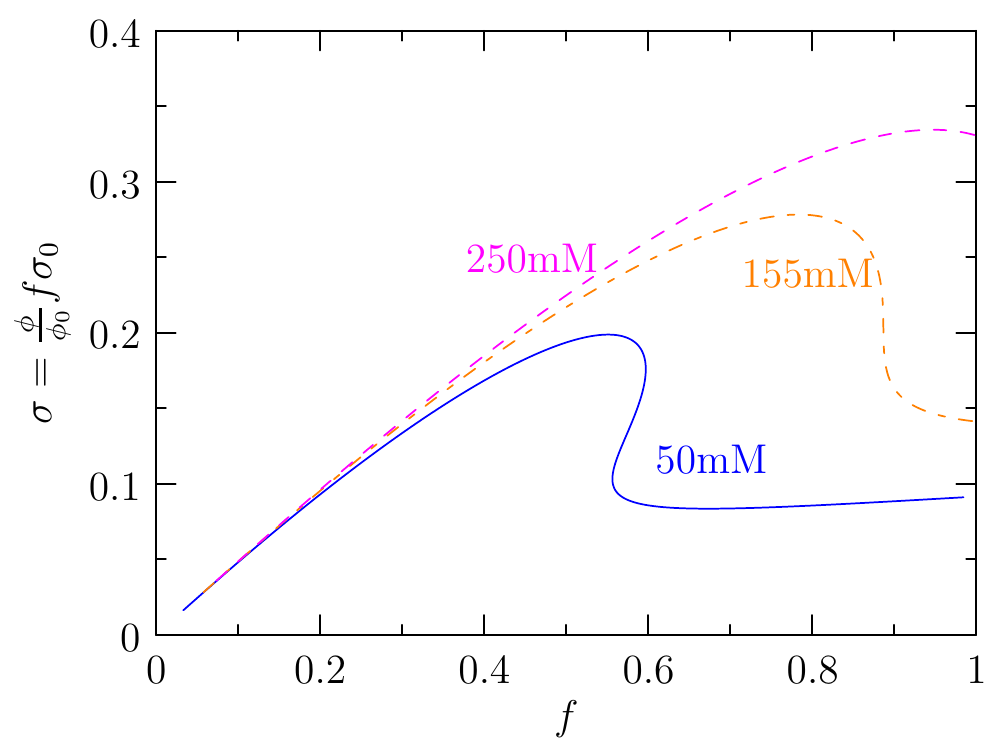}
\hspace{.4cm}
\includegraphics[scale=.49]{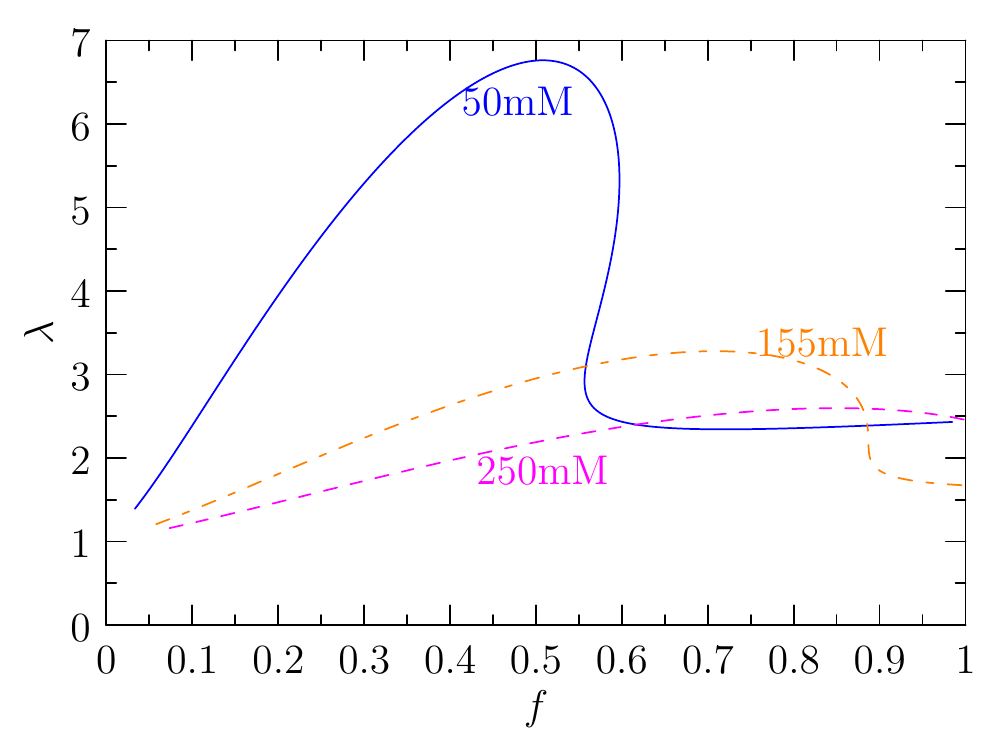}
}
\end{minipage}
\caption{The left, middle and  right figures show the {\it uniaxial swelling ratio} $\alpha$,  the {\it fixed charge density} $\sigma$ and the {\it Donnan ratio} $\lambda$ respectively, versus {\it ionized fraction}   $f$,  for different salt concentration.}\label{Ronfig1}
\end{figure} As expected, swelling increases with decreasing salt concentration.  For  $\Cnacl'= 250\textrm{mM}$,
the membrane remains in its essentially collapsed state with $\alpha<1$. In this case hydrophobicity is the dominant force, and the effect of ion osmotic stress is relatively weak. For  $\Cnacl'= 155\textrm{mM}$ and $50\textrm{mM},$ an initially shallow relation between $\alpha$ and $f$ is punctuated by a rather
sharp rise, indicating initial dominance of hydrophobicity that is overcome by
ion osmotic stress at higher $f$.  As ionic strength (salt concentration)
decreases, the sharp rise occurs at lower ionization degree.

For  $\Cnacl'= 250\textrm{mM}$ and $155\textrm{mM},$ $ f$ uniquely determines $\alpha$.  For
$\Cnacl'= 50\textrm{mM}, $ however, a range of bistability and hysteresis is observed. 
Over this range in $f,$  total stress $s$ vanishes at three values of $\alpha$,
corresponding to two free energy minima and one maximum in between.  The latter,
which corresponds to the negative slope branch of the swelling curve between the
turning points, is unstable, and need not be considered in the discussion of
equilibria. 
Figure \ref{Ronfig1}b shows the fixed charge density, $\sigma=(\phi/\phi_0)f\sigma_0$, as a function of $f$ for the
three salt concentrations.  In all cases, this quantity initially rises with
increasing $f$,
 since $f$ is increasing but swelling does not change
significantly.  The rise is followed by a drop attributed to the sudden increase in swelling.  Bistability is observed for
$\Cnacl' = 50\textrm{mM} $ over the same interval of $f$ as before.

Figure \ref{Ronfig1}c exhibits the calculated Donnan ratio $\lambda$ as a function
of $f$ for the three salt concentrations. 
 While these ratios provide information regarding the ion osmotic swelling force via (\ref{ras14}), they
also provide a link between external and internal pH. The curves are all
nonmonotonic, with bistability for $\Cnacl'= 50\textrm{mM}.$  The swings in $\lambda$
more or less follow those of $\sigma$ and can be
explained in essentially the same way.
 
Figure \ref{Ronfig2} displays the relationship between pH and swelling for the
three salt concentrations, considering both internal and external pHs (dotted and continuous lines, respectively).  These
values are determined according to (\ref{ras13}) and  (\ref{ras15}), with $s=0$, and $pH(int)=-\log\Ch$ and
$pH(ext)=pH(int)+\log\lambda$.    Evidently, increasing salt concentration leads
to an alkaline shift (higher pH) at which the major transition in swelling
occurs.  With increasing salt, a larger degree of ionization and hence higher pH
is needed to effect the transition.  The three pairs of graphs exhibit
qualitatively different behaviors.  For $\Cnacl'= 250\textrm{mM},$ both curves are
monotonic and single valued, although internal pH is clearly lower than external
pH, as expected based on Donnan partitioning of hydrogen ion.  For $\Cnacl' =
50mM,$  both curves exhibit bistability.  At the intermediate salt
concentration, $\Cnacl' = 155\textrm{mM}$, the internal pH curve is single valued, while
the external pH curve shows bistability.  This difference can be attributed to
the Donnan effect, which nudges external pH to higher values at low degrees of
ionization and swelling.  Of course, the real control variable is external pH.

We conclude that there are two potential mechanisms underlying bistability, one
originating from the relation between fixed charge density and degree of
swelling, and the other stemming from the Donnan effect.  While in these studies
we looked at qualitative behaviors as a function of  $\Cnacl'$, these behaviors
might also be observed by altering structural parameters of the hydrogel.
\begin{figure}
\centerline{
 {\includegraphics[scale=.6]{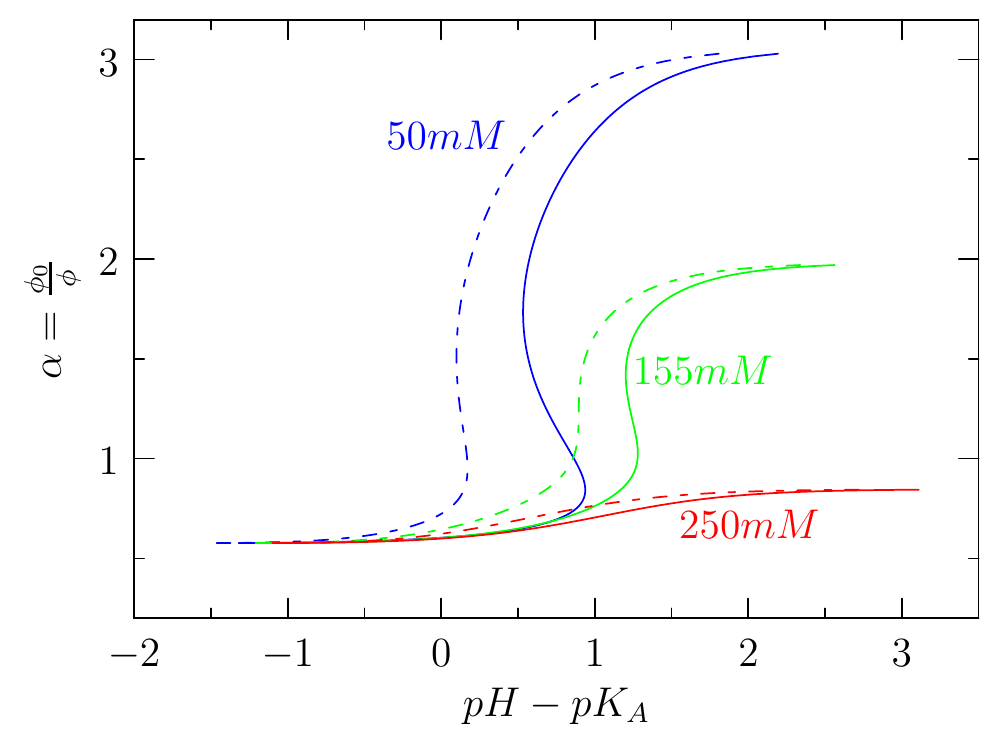}}}
 \caption{\small{Swelling Ratio $\alpha$}  versus  $pH-pK_A$ at different salt concentrations. }\label{Ronfig2}
\end{figure}
To simplify the notation, from now on, we will remove the 'prime' symbol from  the external salt concentration  variable  $C'_{\textrm{\tiny{NaCl}}}$
 and agree in writing $\cnacl$ instead.
 \section{Spatially homogeneous chemomechanical model}{\label{chemo-mechanical model}}
 Now that the axially restricted equilibrium swelling properties of the hydrogel membrane have been explored, we introduce a lumped parameter model of the chemomechanical  oscillator illustrated in Figure 1.
 Variants of this model have been presented previously \cite{dhanarajan2002autonomous, dhanarajan2006}. The dynamic model is based on the following assumptions:
  
\noindent
{\bf a.\,} Membrane properties, including the fixed charge concentration,  swelling state and ion concentrations
inside the membrane are assumed to be homogeneous in space. Upon suppressing lateral swelling and shrinking, the  swelling state of the membrane is determined by its thickness $L=L(t)$.   This homogeneity assumption is made even though the membrane is subjected to a pH gradient.  External pH is "homogenized" by averaging the constant $\textrm{H}^+$ concentration in Cell I, $C_H^I$ and the variable $\textrm{H}^+$ concentration in Cell II, $C_H^{II}$,  that is, $pH=\log[(C_H^I+C_H^{II})/2]$.

 \noindent
{\bf b.\,} Chemical potentials inside the membrane are determined by the
1-D FRDL equation of state for one-dimensional swelling, as described above.
From these, the total 
swelling stress, including the elastic, mixing and ionic contributions, is given by (\ref{ras15}), with $\lambda$
determined according to the electroneutrality condition (\ref{ras13}).  
Electroneutrality is enforced by rapid exchange of ${\textrm{Na}}^+$ and ${\textrm{Cl}}^-$ between the membrane and Cells I and II. 
$\textrm{NaCl}$ concentrations in Cells I and II are assumed constant and equal \cite{Hirotsu}.

\noindent
{\bf {c.}\,} Enzymatic conversion of glucose to gluconic acid is assumed to be instantaneous. This assumption is valid when the concentration of enzyme in Cell II is sufficiently high that transport of glucose into Cell II is
rate limiting. Also, gluconate and bicarbonate, produced according to chemical reactions (I) and (II), respectively, are presumed to not perturb the system's dynamics.  The latter assumption is expected to hold best during early oscillations.

\noindent
{\bf {(d)}\,} The rate of change of fixed, negative charge concentration is controlled by
the rate of transport of hydrogen ions, which reversibly bind to pendant carboxylates as they
diffuse through the membrane \cite{grimshaw1990kinetics}.
 
 \noindent
 {\bf {(e)}\,} Permeability of the membrane to glucose is expressed as 
 $\kgzero e^{-\beta\phi}$, as suggested by free volume theory \cite{LamazePeterlin1971}, \cite{YasudaLamazePeterlin1971}. The parameters $\kgzero$ and $\beta$ represent,
respectively, the hypothetical permeability with vanishing polymer
concentration, and the sieving effect of the polymer, which depends on both
polymer chain diameter and radius of the diffusant (glucose in this case).
For hydrogen ions or water, which are much smaller than glucose, the
sieving factor is assumed to be negligible, and we simply multiply the
respective permeability coefficients, $\kh$ and $\kw$,  by $1- \phi$ to account for the aqueous space in the hydrogel that is available for transport of solutes.
\begin{remark} In this model, diffusion of $\textrm{Na}^+$ and $\textrm{Cl}^-$ are regarded as instantaneous, while diffusion of
 $\textrm{H}^+$ is regarded as rate determining, even though the diffusion coefficient of $\textrm{H}^+$ is decidedly larger.  This assumption is 
justified in part due to the reversible binding of $\textrm{H}^+$ to the pendant carboxylates, which does not occur with $\textrm{Na}^+$ and 
$\textrm{Cl}^-$, and partly due to the very low $\textrm{H}^+$ concentration, which qualifies it for "minority carrier" status.  Changes in concentration of $\textrm{H}^+$ in the membrane are rapidly "buffered" by readjustments of $\textrm{Na}^+$  and $\textrm{Cl}^-$ concentrations, preserving electroneutrality.
\end{remark}

Based on these assumptions, we write the following differential equations  for $L$, the flux of hydrogen ions into the
membrane, which then become either free intramembrane hydrogen ions 
or protons bound to pendant carboxyls, with respective concentrations $\Ch$ and $\Cah$, and the flux of $\rm{H}^+$ to cell II:
\begin{eqnarray}
 \frac{dL}{dt}=&& -\kw(1-\phi)[ln(1-\phi)+\phi +(\chi_1+\chi_2\phi)\phi^2 +\nuw\rho_0(\frac{\phi_0}{\phi}-\frac{\phi}{2\phi_0})\nonumber\\
 &&\,\,-\nuw\Cnacl(\lambda+\frac{1}{\lambda}-2)], \label{d1}\\
 \frac{d}{dt}[L(\Ch+\Cah)]&&=2\kh{(1-\phi)}[\lambda\frac{(\ChI+\ChII)}{2}-\Ch],
\label{d2}\\
 \frac{d}{dt}\ChII&&= \frac{A\kgzero}{V}e^{-\beta\phi}\Cg- 
\frac{A\kh}{V}{(1-\phi)}(\lambda\ChII-\Ch)-\kmar\ChII, \label{d3}
  \end{eqnarray}
where $V$ is the volume of Cell II and $\Cg$ is the glucose concentration in Cell I.  Then, the governing system consists of equations (\ref{d1})-(\ref{d3})  together with the algebraic constraints (\ref{mass-balance-1d}), (\ref{ras9}), (\ref{d4}) and (\ref{ras13}).  With an additional scaling argument, we obtain  the equations to analyze. 
\section{Model scaling and  the governing system}{\label{scaling}}
In this section, we identify the relevant parameters of the system and their numerical values.  This will allow us to rigorously derive a reduced model from the system of equations (\ref{d1})-(\ref{d3}). 
Five dimensionless groups give the system multiscale structure and properties, such as the role of the lower dimensional manifold discussed in section \ref{inertial-manifold}.
 However, the full dynamics cannot be explained in terms of the dimensionless parameter groups only. We find that the range of the oscillatory behavior is further determined by a few individual parameters, such as $\cnacl$, $\sigma_0$ and $\phi_0$, with the remaining ones held fixed. These values are in full agreement with the experiments and are also those used in the simulations of sections \ref{swelling-equilibria} and \ref{numerical-simulations}.

For convenience, we relabel the variable fields of the problem and define dimensionless variables:
 \begin{eqnarray}
 && x=\Ch, \,\, y=\Cah, \,\, z=\ChII, \,\, h=\ChI, \label{fieldsnew}\\
&&\bar x=\frac{x}{c}, \quad \bar y=\frac{y}{c}, \quad \bar z=\frac{z}{c}, \quad \bar L=\frac{L}{L_0}, \quad \bar t=\frac{t}{T},  \label{scaled-variables}
\end{eqnarray}
where $T$ and $c$ are typical time and concentration variables, respectively. 
For now, we use the superimposed {\it bar} notation to represent 
dimensionless quantities.
We choose
\begin{equation}
 T=\frac{L_0\phi_0}{\kw}, \quad c=\Ka.
\end{equation} 
Note that $T$ corresponds to the time scale of equation (\ref{d1}), which will   allow us to properly separate the dynamics of the membrane from that of the chemical reactions.  The choice of  $c$ will lead to  a reduced model resulting from the simplification of equation (\ref{d2}) as we shall see later in the section.  

We now introduce five relevant dimensionless parameter groups consistent with the proposed scaling:
\begin{gather}
\mathcal A_1=\frac{A}{V}\frac{\Cg}{\Ka}\Kg^0 T, \quad \mathcal A_2=\frac{A}{V} T \Kh, \quad \mathcal A_3=\kmar T,  \quad \mathcal A_4= 2\frac{\Kh}{K_w}\phi_0\frac{\Ka\phi_0}{\sigma_0}, \quad \mathcal A_5=\frac{\Ka\phi_0}{\sigma_0}. \label{A}
\end{gather}
The scaled form of equation (\ref{d4}) combined with (\ref{ras9}) is then
 \begin{equation} \bar y=\mathcal A_5^{-1} \frac{\phi\bar x}{1+\bar x}, \,\,\, {\textrm{so}} \,\,\, \,
\bar x+\bar y=\bar x(1+ \mathcal A_5^{-1}\frac{\phi}{1+\bar x})=\mathcal A_5^{-1}\frac{\bar x\phi}{1+\bar x}+O(1). \label{x+y}
\end{equation}
The latter approximation holds for the range of parameters of the  problem, for which $
\mathcal A_5^{-1}\approx 3*10^4. $ Applying it to 
 the scaled version of equation (\ref{d2}), which together with (\ref{d1}) and (\ref{d3}), give the set of equations of the model that we analyse: 
\begin{eqnarray}
\frac{d\phi}{d\bar t}=&& \cR_1,\label{dphi}\\
\frac{d\bar x}{d\bar t}=&&\sA_4\cR_2+ \mathcal A_5(1+\bar x)^2 \cR_1, \label{dx}\\
\frac{d\bar z}{d\bar t}=&& \mathcal A_1\cR_3, \label{dz}\\
&&\lambda=p(\phi)f+\sqrt{p^2(\phi)f^2+1},\quad f=(1+{\bar x})^{-1}, \label{electroneutrality-soln}
\end{eqnarray}
where
 \begin{eqnarray}
\cR_1(\phi,\lambda)=&& (1-\phi)\phi^2\big( H(\phi) -R(\lambda)\big), \quad 
\cR_2(\phi, x, \lambda)=(1-\phi)(1+x)^2(\frac{\lambda}{2}\bar z-\bar x), \label {R2}\\
\cR_3(\phi, x, z, \lambda)=&&  e^{-\beta\phi} -\mathcal A_2\sA_1^{-1}{(1-\phi)}(\lambda \bar z-\bar x)- \mathcal A_3\sA_1^{-1} \bar z, \,\,\,\, \textrm{and} \label{R3}\\
H(\phi)= && \ln(1-\phi)+
\phi+(\chi_1+\phi\chi_2)\phi^2+\bar\nuw\bar\rho_0(\frac{\phi_0}{\phi}-\frac{\phi}{
2\phi_0}), \label{Hlambda}\\
R(\lambda)= &&  \nuw \cnacl(\sqrt{\lambda}-\frac{1}{\sqrt{\lambda}})^2, \quad p(\phi)= \frac{\gamma\phi}{2(1-\phi)}, \label{Rlambda}\\
q(\phi)= &&1+\frac{1}{2\gamma_0}H(\phi), \quad 
\gamma=\frac{\sigma_0}{\cnacl\phi_0}, \quad \gamma_0=\vw \cnacl. \label{gamma-gamma0}
\end{eqnarray}
The graph of $H(\phi)$ is shown in figure \ref{Hfig}.

Parameters of the problem, shown in tables (S2)-(S4) of the Supplementary Materials section, are of three main types: device specifications, hydrogel parameters and rate quantities. It should be noted that while some of
these values either reflect the experimental conditions or are literature
parameters for NIPA/MAA hydrogels \cite{PEN}, some parameters were chosen to
produce results that are in line with experimental observations. 
Specifically, 
the size of the marble component in Cell II was chosen with  a sufficiently large surface area so as to provide the necessary reaction sites to remove excessive hydrogen ions from the system. For the given parameter values, we find that 
\begin{equation} \mathcal A_1\approx 0.16*10^{1}, \,\, \mathcal A_2\approx 0.52* 10^{-2}, \,\, \mathcal A_3\approx 0.75*10^{-2}, \,\, \mathcal A_4\approx 0.33 *10^{-5}. \label{parameter-values} \end{equation}
We point out that $
0<\mathcal A_4<< 1, \quad \mathcal A_1^0:= A_1e^{-\beta\phi_0}= 0.016\,\,\,\, \textrm {and} \,\,\, \, \mathcal A_1^0>\mathcal A_2, \mathcal A_3. \,$  \, 
So, the  first term of the  right hand side of  (\ref{dz}) is of the order  $\mathcal A_1^0$, since
\begin{equation}
\mathcal A_1\mathcal R_3= \mathcal A_1^0 e^{-\beta(\phi-\phi_0)} + \dots \label{eqn-dz0}
\end{equation}
These parameter properties are relevant to obtain estimates of the solutions.
\begin{definition}\label{parameter-space}
 We let $\mathcal P$ denote the set of parameters in tables (S2)-(S4) and with
\begin{equation}
\cnacl\in [20*10^{-3}, 155*10^{-3}] \textrm{M},  \quad \sigma_0= [0.10, 0.40], \quad \phi_0=[0.1, 0.3]. \label{parameters}
\end{equation}
\end{definition}
Experiments have been carried out with values of $\cnacl$ in the list
$ \{350, 250, 150, 100, 50, 40, 25, 20\}*10^{-3}$M. However for salt concentration values  above $155*10^{-3}$M, no oscillatory behaviour has been found. 

From now on, we suppress the superimposed {\it bar} in the equations and assume  that all variables are scaled. 
We now study the properties of the functions in the governing system.  
 \begin{figure}\label{Hphi}
 \begin{minipage}[t]{\textwidth}
\centerline{
\includegraphics[scale=.55]{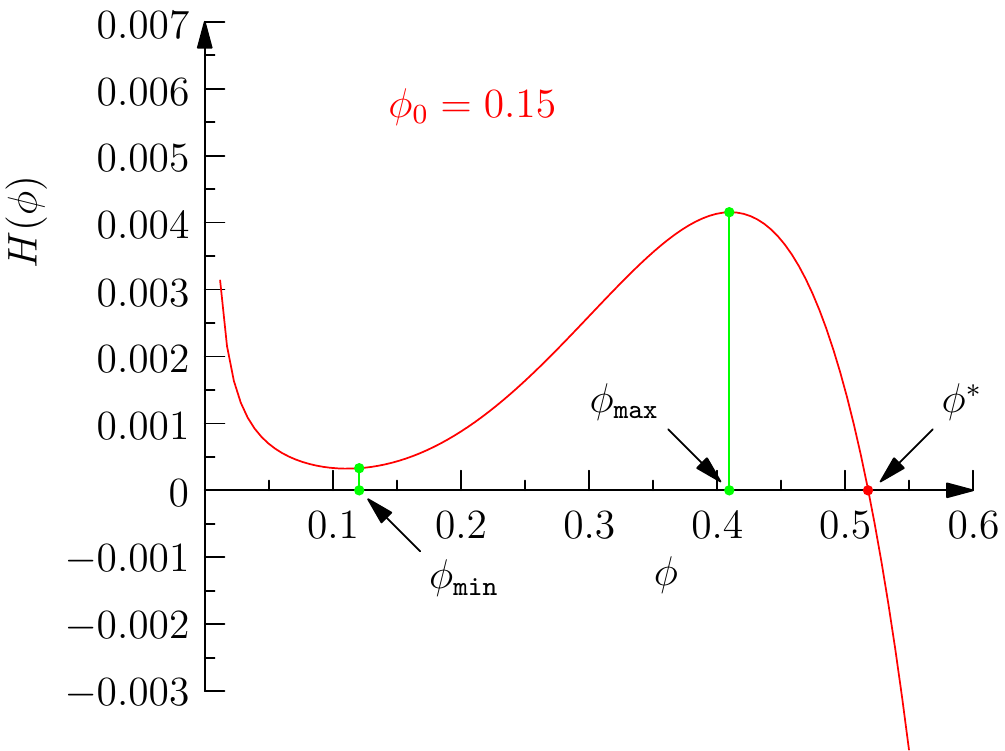}
\hspace{.50cm} 
\includegraphics[scale=.55]{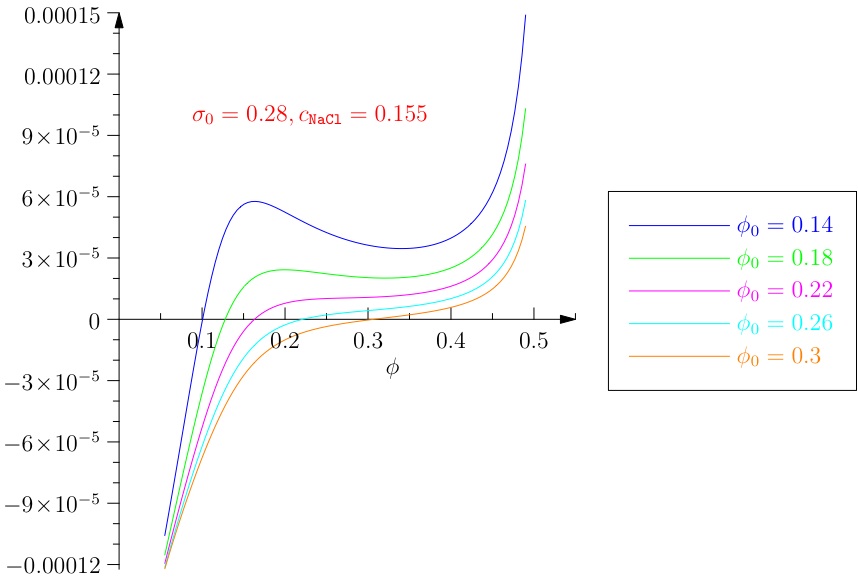}}
\end{minipage}
\caption{\small{(Left) Plot of  the mechanical compliance $H(\phi)$ for  parameters values: $\chi_1=.48, \chi_2=.6$ and
$\phi_0=.3$;  $\phi^*=0.5226$.  (Right) Plot of the left hand side function of equation (\ref{equilibrium-phi})}}\label{Hfig}
\end{figure}
\begin{lemma}
Suppose that the parameters in $(\ref{Hlambda})$ belong to $\mathcal P$.  Then the  function $H(\phi)$  is non-monotonic in $(0,1)$ and has the following properties:
\begin{equation}
 \lim_{\phi\to 0^{+}}H(\phi)=\infty, \quad H(\phi)=O(\frac{1}{\phi}) \,\, {\textrm {and}}\, \, 
 \lim_{\phi\to 1^{-}}H(\phi)=-\infty, \quad H(\phi)=O(\log{(1-\phi)}).
\end{equation}
Moreover, there exists a unique $\phi^*\in(0,1)$ such that 
$
H(\phi^*)=0.$ 
Furthermore,  $H(\phi)$ has a unique local maximum,  $\phi_{\textrm{\begin{tiny}{max}\end{tiny}}}$  and a unique local minimum,  $\phi_{\textrm{\begin{tiny}{min}\end{tiny}}}$ in $(0,\phi^*)$, as shown in Figure \ref{Hphi}.
\end{lemma}

We represent the governing system (\ref{dphi})-(\ref{dz}) and denote the corresponding domain as
  \begin{equation}
\dot{\mathbf x}= \boldsymbol f(\bx) \quad \textrm{and} \quad \mathcal I=\{\bx=(\phi, x, z): \, \phi\in(0, 1), x\in(0, \infty), z\in(0, \infty)\}, \label{I}
\end{equation}
respectively, with the components of $\boldsymbol f$ given by the right hand side functions of the system. It is easy to check that $\boldsymbol f$ is continuously differentiable in the convex set $\mathcal I$, also the set of initial data of the problem. In order to  denote  solutions corresponding to specific initial data in $\mathcal I$, we write
\begin{eqnarray}
&&\phi(t)= \phi(\phi^0, x^0, z^0)(t), \, \, x(t)= x(\phi^0, x^0, z^0)(t), \,\, z(t)=z(\phi^0, x^0, z^0)(t), \nonumber\\
&&\phi_M(t)= \phi(\phi^0, x^+(\phi^0), z^0)(t), \, \, x(t)= x(\phi^0, x^+(\phi^0), z^0)(t), \,\, z(t)=z(\phi^0, x^+(\phi^0), z^0)(t), \nonumber
\end{eqnarray}
with the latter expression referring to restrictions to the two-dimensional manifold $\mathcal M$ defined in  (\ref{Mminus}).

\subsubsection{Time scales of the model}
The  relative sizes of the  dimensionless parameters $\mathcal A_i$, in particular,  the fact that $\mathcal A_4<< 1$  in equation (\ref{dx}) indicates that   $x$  is the slow field of the problem.  Accordingly, we define the time scale associated with $x$ as
\begin{equation}
\tau=\sA_4 \bar t= \frac{\sA_4}{T} t, \quad  \epsilon:=\sA_4, \quad \mu=\frac{\mathcal A_4}{\mathcal A_1}. \label{time-scales}
\end{equation}
So, the original dimensionless time variable $\bar t$ corresponds to the {\it fast } dynamics, whereas $\tau$ gives the {\it slow } time, that now we take as reference.
In the new time scale, the governing system reduces to 
\begin{eqnarray}
\cA_4\frac{d\phi}{d\tau}=&& \cR_1(\phi, \lambda)\label{phi333}\\
\frac{d\bar x}{d\tau}=&&\cR_2(\phi,\bar x,\lambda)+ \frac{\Ka\phi_0}{\sigma_0\mathcal A_4}(1+\bar x)^2\mathcal R_1 , \label{x333}\\
\frac{\cA_4}{\cA_1}\frac{d\bar z}{d\tau}=&& \cR_3(\phi, \bar x, \bar z,\lambda), \label{z333}
\end{eqnarray}
together with equations  (\ref{electroneutrality-soln}).  
For our typical parameter values, $\mathcal A_4=0.33*10^{-5}$, $\frac{\mathcal A_4}{\mathcal A_1^0}=0.12*10^{-3} $  and $\frac{\Ka\phi_0}{\mathcal A_4\sigma_0}=10.$
This motivates identifying a reduced model  consisting of equations (\ref{x333})-(\ref{z333}) together with the algebraic equation
$
0=\cR_1(\phi, \lambda),
$
and equations  (\ref{electroneutrality-soln}).
We also point out  that the scaling falls through when $\phi>0$ is arbitrarily small, that is, for highly swollen regimes. This follows from estimating equation (\ref{dphi}) for $\phi>0$ small as well as the observation that $\lambda=O(1) $ and $R(\lambda)\approx 0$. Specifically, 
 for  given $\varepsilon>0$, and for $\phi\in(0, \varepsilon)$:
\begin{equation}
 \label{asymptotic}
p(\phi)=\gamma\phi+o(\varepsilon),\quad
\lambda=1+ o(\varepsilon), \quad R(\lambda)= o(\varepsilon), \quad H(\phi)=\vw\rho_0\frac{\phi_0}{\phi}+ o(\varepsilon). 
 \end{equation}
 This allows us to establish the following lemma:
\begin{lemma}
Let $\tilde T>0$ be such that $\phi=\phi(t)$ and $\lambda=\lambda(t)$ satisfy equation (\ref{dphi}) for $t\in [0,\tilde T]$. Then
\begin{equation}
\frac{d\phi}{dt}= \frac{\nu_w\rho_0}{\cA_4}O(\frac{1}{\phi})\,\, {\textrm{as}}\,\, \phi\to 0^+.
\end{equation}
\end{lemma}
The behavior of the constitutive functions in (\ref{asymptotic}) yields the following: 
\begin{proposition} \label{boundedness} Solutions of (\ref{dphi}) corresponding to initial data in $
\mathcal I$ have the property  that $\phi$ remains bounded away from $\phi=0$ and $\phi=1$ for all time. Moreover, there exist positive numbers $t_{\textrm{\begin{tiny}m\end{tiny}}}$  and $t_{\textrm{\begin{tiny}M\end{tiny}}}$,  depending on the initial data,   
such that $\phi(t_{\textrm{\begin{tiny}m\end{tiny}}})>0$ and $\phi({t_\textrm{\begin{tiny}M\end{tiny}}})>0$ are minimum and maximum values of  $\phi$, respectively.
\end{proposition}
 \begin{proof}Let us first recall that $\phi^*$ denotes the largest value of $\phi$ such that $H(\phi^*)=0$.  From the governing equation for $\phi$, we see that for initial data $\phi_i\geq \phi^*$, $\frac{d\phi}{dt}<0$. This proves the boundedness of the orbit away from $\phi=1$. 
 Let us consider an orbit with initial data $\phi_i\in(0,\eps)$ and such that $\frac{d\phi}{dt}<0$ at some $t>0$.   Integrating the governing equation for $\phi$ while taking into account the
 estimates in (\ref{asymptotic}) on its right hand side function yields
 $\phi(t)=\phi_i O(e^{\frac{t}{\vw\rho_0\phi_0}}).$
 This contradicts the assumption that $\frac{d\phi}{dt}<0$ at some $t>0$, and so proving the statement of the proposition.
   \end{proof}
     \begin{remark}
   Proposition (\ref{boundedness}) states that the solution $\phi$ of the governing system remains in the interval $(0, 1)$, for all time  of existence. However, due to the separation of time scales, only disconnected subintervals of $(0,1)$  are admissible.  This is the topic of the next section. 
   \end{remark}
   
      The following lemma is used in the estimate of bounds of solutions. Let us  rewrite equation (\ref{dz}) as  
\begin{equation}
\frac{dz}{dt}+w(t)z=g(t)\, \, \textrm {\,with\,}\,  w:= \cA_2(\lambda(1-\phi)+ 1), \quad g:=\cA_1e^{-\beta\phi}+\cA_2(1-\phi)x.\label{z-auxiliary}
\end{equation}
    The proof follows by integration of the previous auxiliary equation.  
\begin{lemma}
Suppose that $z_0>0$,  $x=x(t)$, $\phi=\phi(t)$ and $\lambda=\lambda(t)$ are prescribed, continuous functions with $ t\in[0, \hat T]$, for some $\hat T>0$. Then, the solution of (\ref{dz}) satisfies the equation
\begin{equation}
z(t)= E^{-1}(t)\big(z_0+ \int_0^t E(s)g(s)\,ds\big), \,\, t\in [0, \hat T], \,z_0=z(0), \quad {\rm{and}}\,\, E(t):= \exp{\int_0^t}w(s)\,ds.
\end{equation}
\end{lemma}

Next, we show that solutions of (\ref{dx}) remain bounded away from $x=0$. 
\begin{lemma}
Suppose that $\phi=\phi(t)$  and $z=z(t)>0$ are prescribed continuous functions for $t\geq 0$, and such that $\phi(t)$ and $z(t)$ remain bounded away from $\phi=0$ and $z=0$, respectively. 
Then the solution $x=x(t)$ of equation (\ref{dx}) corresponding to initial data $x(0)>0$ remains bounded away from  $x=0$ for all $t>0$. 
\end{lemma}
\begin{proof}
We argue by contradiction and assume that for a prescribed arbitrarily small $\epsilon>0$,  there exits $t_1=t_1(\epsilon)>0$ such that $0\leq x(t)\leq x(t_1)\leq\varepsilon$, 
for $t\geq t_1$.  A simple calculation using 
(\ref{electroneutrality-soln}) gives 
\begin{equation}
\lambda(t_1)= p(\phi(t_1))+ \sqrt{p^2(\phi(t_1))+1}+ O(\varepsilon).
\end{equation} 
Hence, $\mathcal R_2(t_1)\geq (1-\phi(t_1))(1+x(t_1))^2(\frac{\lambda(t_1)}{2}z(t_1)-\varepsilon).$
So, $\frac{dx}{dt}(t_1)>0$ and bounded away from 0, and therefore, $x(t)$ cannot further decrease to 0. 
\end{proof}

Using the previous lemmas on boundedness of solutions,  we can now state the following.
\begin{proposition} \label{bounded-orbits}
Suppose that the parameters of the governing system belong to $\mathcal P$.  Let $\{\phi=\phi(t), \,x=x(t)$,  $z= z(t)\}$, $t\in [0, T)$, denote a solution of the 
system (\ref{dphi})-(\ref{dz}) and (\ref{electroneutrality-soln}) corresponding to initial data in $\mathcal I$, and with $T>0$ representing the maximal time of existence. 
 Then $\{\phi(t), x(t), z(t)\} $ are bounded. Moreover, the lower bounds are strictly positive and  the upper bound of $\phi$ is strictly less than 1.  Furthermore, $\mathcal I$
is invariant under the flow of the governing system.  
\end{proposition}
\section{ Hopf bifurcation: a numerical study}{\label{numerical-simulations}}
  We numerically investigate the following properties of the solutions corresponding to the parameter set $\mathcal P$: 
\begin{enumerate}
\item Uniqueness of the steady state solution and its stability.
\item Occurrence of  Hopf bifurcation.
\item Non-monotonicity of the graph $x=x^+(\phi)$  in (\ref{lambdapm}) and  the reduced system (section (\ref{inertial-manifold})).
\end{enumerate}
We have seen that the third property guarantees the construction of  closed orbits of the approximate two-dimensional system. This together with existence of 
 a unique, unstable  steady state provide sufficient conditions for the Poincare-Bendixon theorem to apply, from which existence of  a limit cycle follows. 

 Steady state values of the variables $x$, $z$  and $\phi$ satisfy 
\begin{eqnarray}
&&x=\frac{\lambda}{2} z, \quad z=   \cA_1e^{-\beta\phi}(\cA_2(1-\phi)\frac{\lambda}{2}+\cA_3)^{-1}, \nonumber\\
&& \frac{p}{\sqrt{q^2-1}}-1- \frac{\cA_1}{2}\lambda e^{-\beta\phi}(\cA_2(1-\phi)\frac{\lambda}{2}+\cA_3)^{-1}=0, \label{equilibrium-phi} 
 \end{eqnarray}
where the last equation follows from the previous expressions of $x$ and $z$, together with  relations (\ref{lambdapm})  and (\ref{gamma-gamma0}). 
The second graph in Figure  \ref{Hphi}
 shows plots of the function on the left hand side of equation (\ref{equilibrium-phi}) with the corresponding unique root,  which after substitution into equations 
for $x$ and $z$,  yields their  steady state values. The computational results summarized next are in full agreement with those of laboratory experiments.

\noindent {{\bf 1.}   For $ \cnacl=0.155$ M:\,}

{\bf a.\,}We found that for each $\phi_0=0.1, 0.15, 0.2, 0.295$,  there is a unique physically relevant  steady state, that is, with $\phi<\phi^*$. 
We point out that this result is valid for $\phi_0\in [0.1, 0.3]$. This is justified by the implicit function theorem applied to the left hand side function in (\ref{equilibrium-phi}) with respect to each root.
In this proof, we also use the fact that the slope of the tangent line to the graph at the intersection with the $\phi$-axis is non-horizontal, as shown   in figure \ref{Hphi} (Right).

 {\bf {b.}}  The stability of the unique equilibrium solution is represented in the diagrams of figures \ref{Hopf-bifurcation}.  We find a region in  parameters plane determined by the Hopf bifurcation curves,
within which the equilibrium is unstable; outside, the solution is stable. The unstable equilibrium has  a pair of complex conjugate eigenvalues with positive real part and a negative eigenvalue. 
Consequently, the existence of a Hopf bifurcation is guaranteed by the {\it Hopf bifurcation theorem } \cite{Marsden1976} (section 3).

{\bf {c.}\,} The  graph $x=x(\phi)$  of the function (\ref{lambdapm}) is found to be non-monotonic inside the regions determined by the Hopf bifurcation curves in figures \ref{Hopf-bifurcation} and monotonic outside.
Oscillatory behaviour is numerically  found inside these regions. 

\noindent{\bf 2.}
We have also found that decreasing $\cnacl$ gives a qualitative behavior analogous to decreasing $\phi_0$, that is, it promotes a single unstable stationary state  of the system and the occurrence of Hopf bifurcation. 
In particular, the following behaviors have been found:

{\bf {a.}\,} For $\cnacl=0.145$M and $\phi_0=0.3$, the steady state solution corresponds to $\phi= 0.145$. 

{\bf {b.}\,} For $\cnacl=0.148$M and $\phi_0=0.3$, the maximum and minimum values  of $x=x(s)$ occur at $\phi=0.24$ and $\phi=0.26$.   In both, these cases, a Hopf bifurcation has also been found.

{\bf {c.}\,} Hopf bifurcations have also been found for $\cnacl=0.125$M and $\cnacl=0.115$M. 

 {\bf {d.}\,} For $\phi_0=0.3$, $\cnacl=0.155$M and $\sigma_0=0.28$, no Hopf bifurcation occurs. The simulations do not show oscillatory behaviour and  there is at least one stable stationary state. 
\begin{figure} \label{Hopf-bifurcation1}
\begin{minipage}[t]{\textwidth}\centerline
{
\includegraphics[scale=.50]{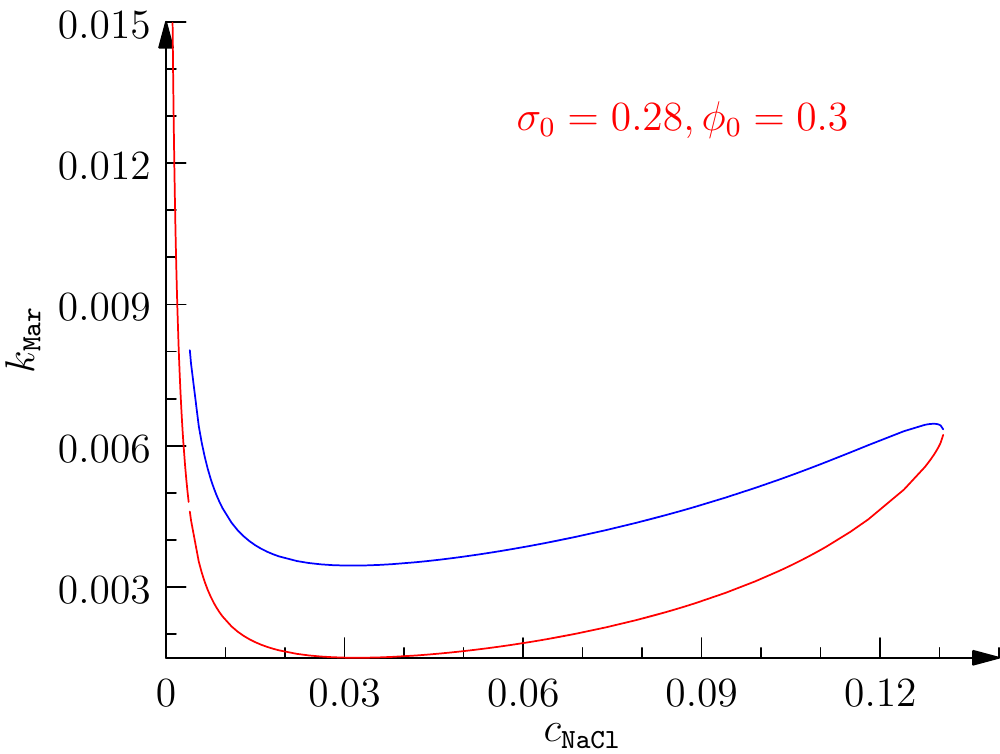}
\hspace{.30cm}
\includegraphics[scale=.50]{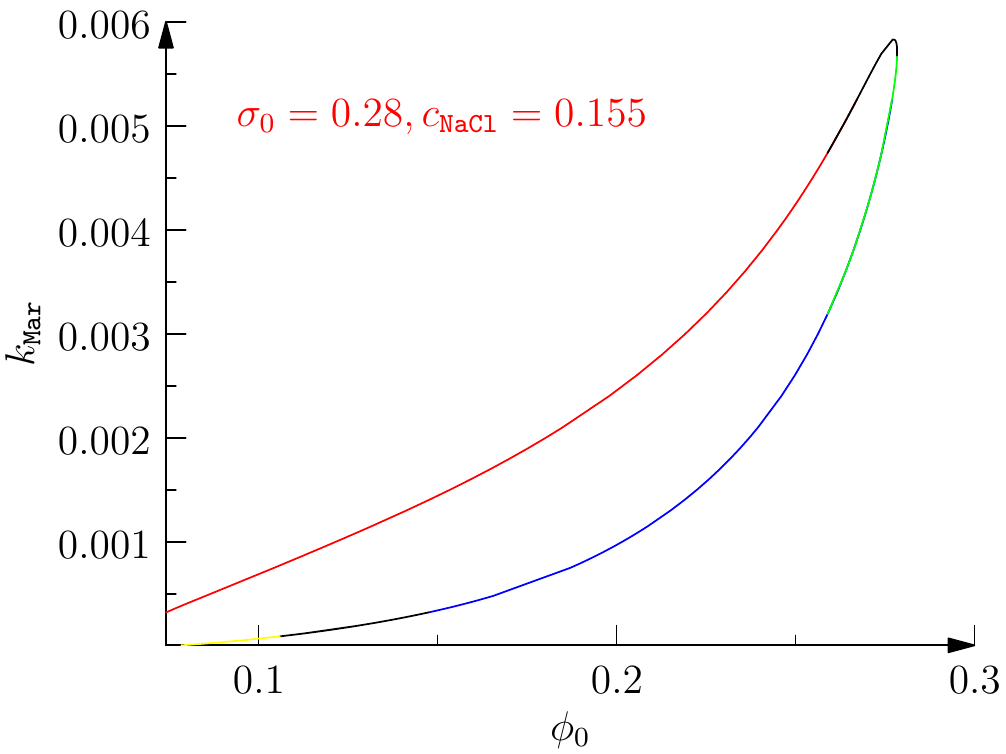}
\hspace{.30cm}
\includegraphics[scale=.5]{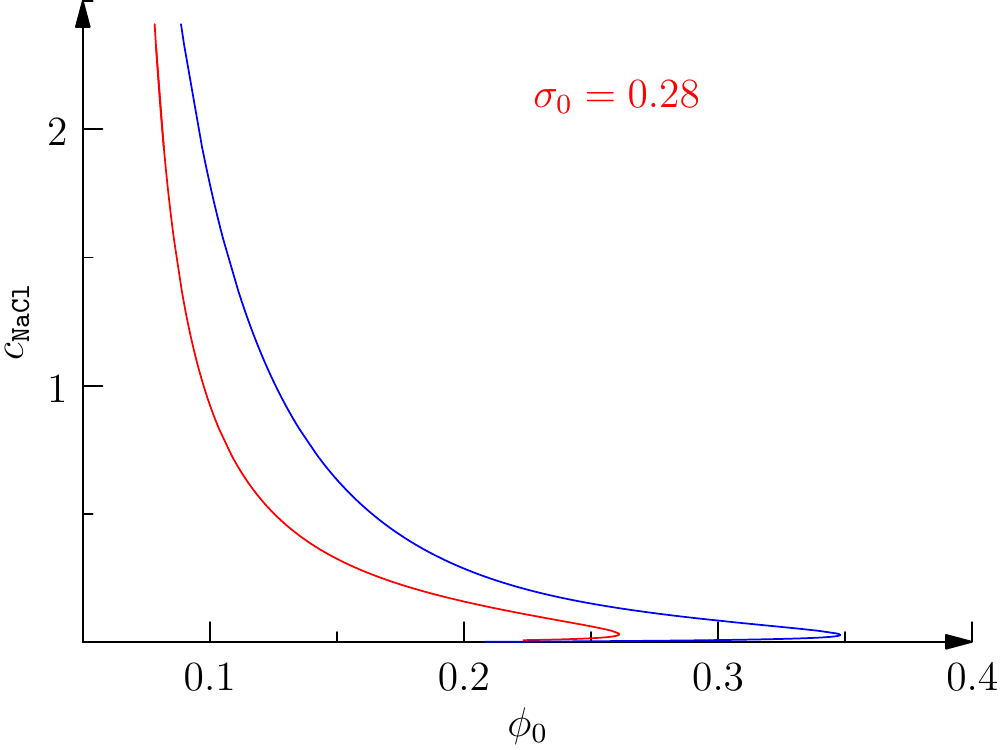}}
\end{minipage}
\caption{Hopf bifurcation diagrams in different parameter planes. The red line denotes the instability threshold. }\label{Hopf-bifurcation}
\end{figure}
\section{ Inertial  manifold of the governing  system}{\label{inertial-manifold}}
We now derive and study a reduced two-dimensional system such that, solutions of initial value problems of the original three-dimensional system remain arbitrarily close to those of reduced one, for most of the time of existence. Moreover, we shall see that, for the parameters of interest, the functions defining the 2-dimensional system are specified in two separate branches. So, 
 it is necessary to extend the concept of solution giving it a weak interpretation as shown in the next section.  Weak solutions also give a good physical description of the phase transition features of the model. 

In order to  identify the slow manifold, we consider the nullcline
\begin{equation}
H(\phi)=R(\lambda)\label{HR}
\end{equation} of the original system.
This corresponds to setting $\epsilon=0$ in (\ref{phi333}), so that $\mathcal R_1\equiv 0$  while keeping $\mu$ fixed. Note that the right hand side of equation (\ref{x333}) is simplified accordingly. 

We now list the properties of the solutions of the equation (\ref{HR}) that follow from the positivity of $R(\cdot)$ and the fact that $\lambda\geq 1$.
\begin{lemma} Let $H(\phi)$ and $R(\lambda)$ be as above. Then, solution pairs $(\phi, \lambda)$ of  equation (\ref{HR})  satisfy:
 $$0<\phi\leq \phi^* \,\, \textrm {so that}\,\, H(\phi)\geq 0 \quad \, \textrm{and} \quad \,
 \lambda>1, \,\, \textrm {with } \, \,\lambda=1\, \, \textrm {when}\,\, \phi=\phi^*.$$
\end{lemma}
This allows us to characterize the nullcline (\ref{HR}) as 
\begin{eqnarray}
\mathcal N=&&\{(\phi, \lambda, z)\in {\mathbf R}^3: 0<\phi<1, \, z>0, \lambda>1,  \, H(\phi)=R(\lambda), \, {\textrm{and such that (\ref{electroneutrality-soln}) holds}}\}, \label{N}\\
\mathcal N:=&& \mathcal N^+\cup \mathcal N^{-},\,\,\,
\mathcal N^+= \{(\phi, \lambda, z)\in \mathcal N: H'(\phi)>0\}, \quad \mathcal N^-= \{(\phi, \lambda, z)\in \mathcal N: H'(\phi)\leq 0\}. \label{N+}
\end{eqnarray}
Since relations (\ref{electroneutrality-soln})  define $\lambda $ as a monotonic function of $x$, we can exchange the roles of $x$ and $\lambda$ in (\ref{N})-(\ref{N+}) at convenience. 

Let us obtain an explicit  representation of $\mathcal N$ and the governing equations of the reduced system. For this, 
we first solve equations  (\ref{electroneutrality-soln}) and (\ref{HR}) using (\ref{Hlambda}) and (\ref{Rlambda})  and  (\ref{gamma-gamma0})  as follows 
\begin{equation}
 \frac{\gamma\phi}{(1-\phi)} f= \lambda-\frac{1}{\lambda}, \quad 
 \frac{H(\phi)}{\gamma_0}= \lambda+\frac{1}{\lambda}-2. \label{lambda-RH}
\end{equation}
Addition and subtraction of these equations yields
\begin{equation}
 \lambda=q(\phi)+p(\phi) f, \quad
\frac{1}{\lambda}=q(\phi)- p(\phi) f,  \label{lambda-p-q}
\end{equation}
respectively,   which, in turn,  give the equation
$\fp=\frac{1}{p}\sqrt{q^2-1},$ {for} $ q\geq 1$ and with $ H(\phi)\geq 0$. 
The corresponding values of the concentration $x$ and the Donnan ratio $\lambda$ are
\begin{equation}
 x^{+}= \frac{p(\phi)}{\sqrt{q^2(\phi)-1}} -1, \quad 
 \lambda^{+}= q(\phi)+ \sqrt{q^2(\phi)-1}, \label{lambdapm}
\end{equation}
respectively. Figures \ref{LdP-XdP} present the graphs of these functions, for  parameters in the class $\mathcal P$, showing their non-monotonicity.  We label the  critical points  of $x^+(\phi)$ as
\begin{equation}
 0<\phi_s^1<\phi_s^2: \,\, \frac{dx^+}{d\phi}(\phi_s^i)=0, \, i=1,2, \quad {\textrm{and}}\quad  x_s^i:=x^+(\phi_s^i),
 \label{crit-points}
\end{equation}
and consider the strictly increasing branches $(0,\phi_s^1)\cup(\phi_s^2, \phi^*)$. Let $\phi^+(x)$ denote the inverse of the restrictions of $x^+(\phi)$ to the monotonic branches, and define
\begin{eqnarray}
\mathcal M_1:=
&& \{(\phi, x, z):  \, z\geq 0, \, x=x^+(\phi), \, \phi\in(0,  \phi_s^1)\}, \quad
 \mathcal M_2
 :=\{(\phi, x, z):  \, z\geq 0, \, x=x^+(\phi), \, \phi\in(\phi_s^2, \, \phi^*)\}, \nonumber\\
 \mathcal M:=&& \mathcal M_1\cup \mathcal M_2, \quad 
 \mathcal M^-:=\{(\phi, x, z):  \, z\geq 0, \, x=x^+(\phi), \, \phi\in(\phi_s^1, \, \phi^2_s)\}. \label{Mminus}
\end{eqnarray}
That is, $\mathcal M$ consists of two  branches where $x^+(\phi)$ increases monotonically. 
 We point out that this graph is the cross section with respect to $z$ of the surface $H(\phi)=R(\lambda)$. This surface divides the whole space into upper and lower regions, with $\cR_1>0$ and $\cR_1<0$, respectively. 

We now study the governing system restricted to $\mathcal M$.
In terms of the original dimensionless time variable (which  we still denote by $t$) and  dependent variables, the governing system in $\mathcal M$
reduces to
\begin{equation}
 \frac{dx}{dt}= \mathcal A_4\phi(1-\phi)(1+\gamma\phi f^2)^{-1}\big(\frac{\lambda}{2} z-x\big), \quad 
 \frac{dz}{dt}= \mathcal A_1e^{-\beta\phi} -\mathcal A_2{(1-\phi)}(\lambda z-x)-\mathcal A_3 z, \label{z-3d}
\end{equation}
together with equations (\ref{lambdapm}).
\begin{figure} 
\begin{minipage}[t]{\textwidth}
\centerline
{
\includegraphics[scale=.40]{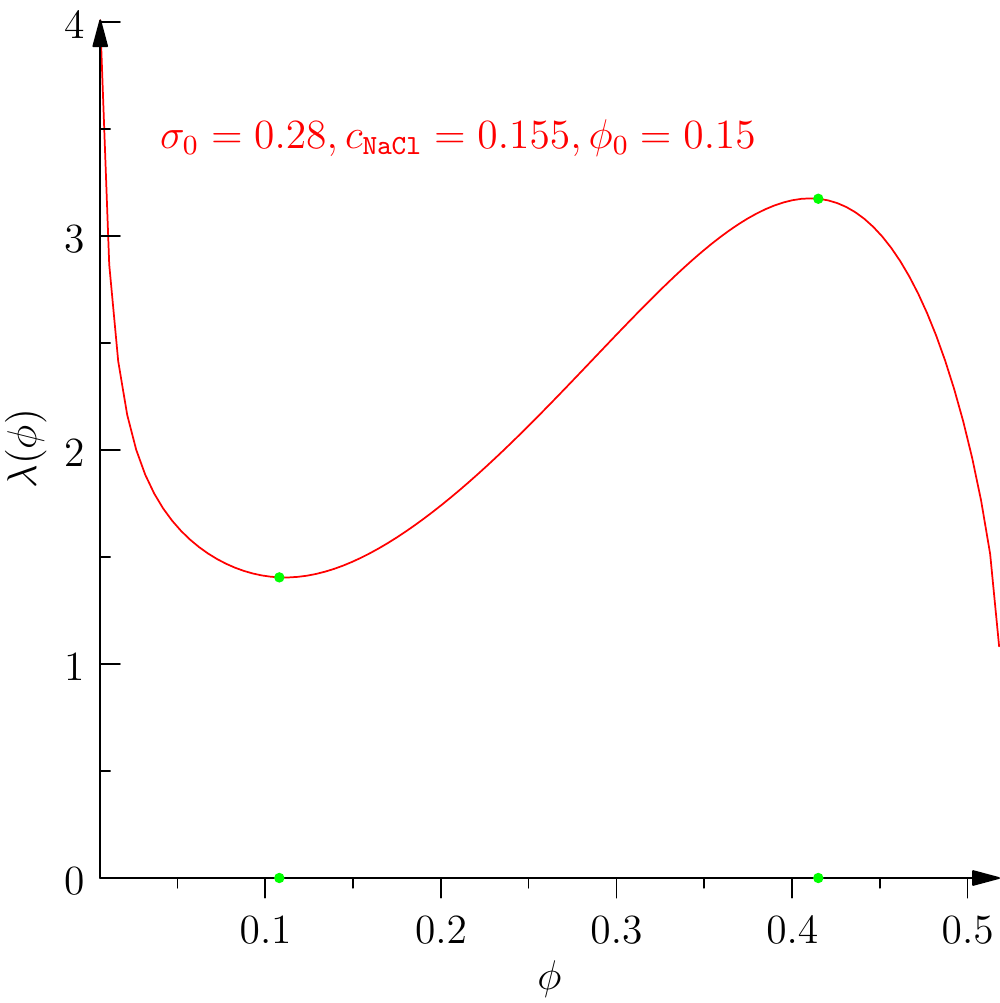}
\hspace{.80cm}
\includegraphics[scale=.40]{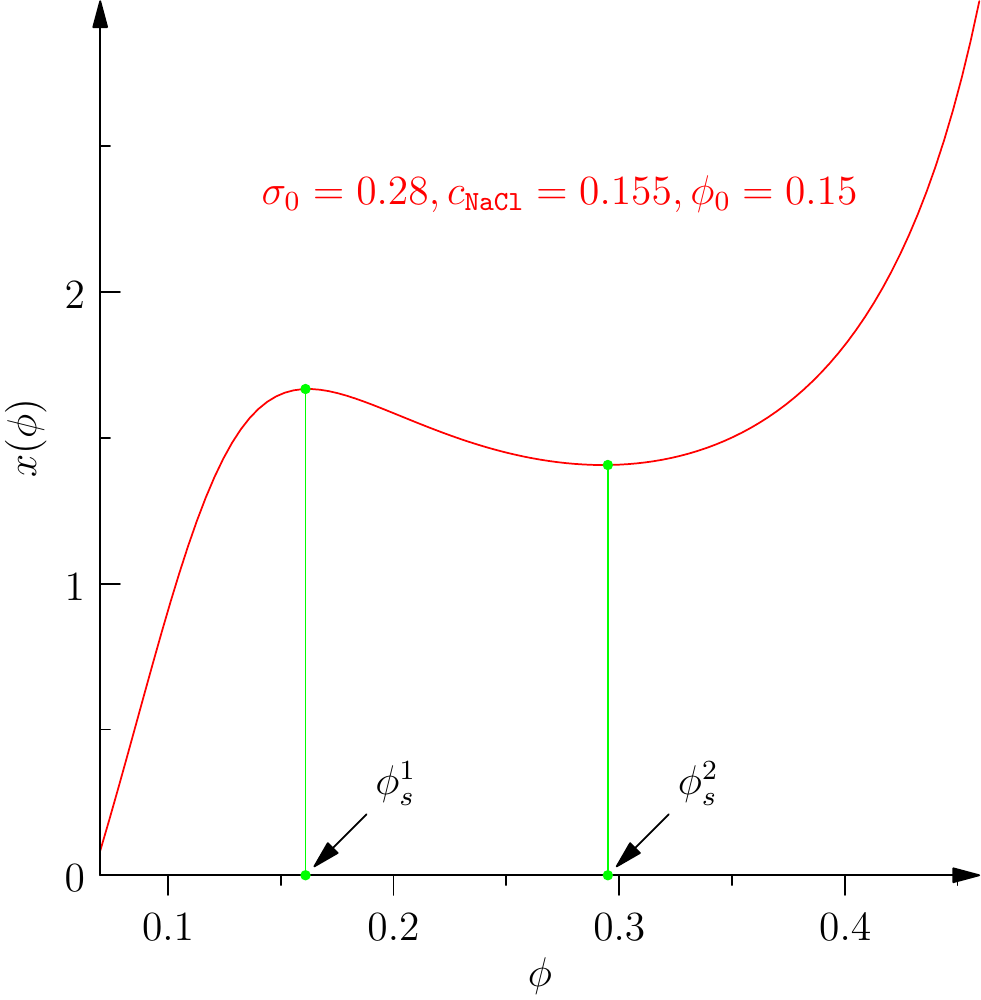}
}
\end{minipage}
\caption{ Plot of $\lambda=\lambda^+(\phi)$  and  $x=x^+(\phi)$ from equations (\ref{lambdapm}) that enter in the definition of $\mathcal M$ }\label{LdP-XdP}
\end{figure}
\begin{proposition} Suppose that the parameters of the problem belong to the class $\mathcal P$. Then $\mathcal M\cup \mathcal M^-$ is a  two-dimensional invariant manifold of the three-dimensional system. Furthermore, the vector field of the two dimensional system is Lipschitz in $\mathcal M$.
 Moreover,
$\mathcal N^-\subset \mathcal M.$ 

\end{proposition}
\begin{proof}
Let us consider initial data $(\phi_0, x_0, z_0)\in\mathcal M\cup\mathcal M^-$. It is easy to see that we can construct a solution of the three dimensional system belonging to $\mathcal M\cup\mathcal M^-$. So, by uniqueness, solutions with initial data satisfying $x_0=x^+(\phi_0)$, $\phi_0\in(0, \phi^*)$ belong to $\mathcal M\cup\mathcal M^-$, for all time of existence.
Moreover, applying  the same arguments as in Proposition \ref{bounded-orbits},  we find  that $x(t) $ remains bounded away from $x=0$, and so $\phi(t)>\phi^{**}$, for all $t>0$.
The last statement of the proposition follows by taking the derivative with respect to $\phi$ of the function $x=x^+(\phi)$ in equation  (\ref{lambdapm}), to give
\begin{equation}
\frac{dx}{d\phi}=\frac{1}{\sqrt{q^2-1}}\big((1-\phi)^{-2}-\frac{pq}{q^2-1}H'(\phi)\big). \label{dxdphi}
\end{equation}
So, $H'(\phi)<0$ implies that $\frac{dx}{d\phi}>0$, for $\phi<\phi^*$.
\end{proof}

In order to study solutions with initial data  outside $\mathcal M\cup \mathcal M^-$, we observe the following properties of the term $\mathcal R_1$:
\begin{equation}
 \mathcal R_1>0  \,\,\textrm{ for } \,\, x>x^+(\phi) \quad \textrm{and}\quad
\mathcal R_1<0 \,\,\textrm{ for } \,\, x<x^+(\phi). \label{R1-sign}
\end{equation}
\begin{proposition}
Suppose that the parameters of the system belong to the class $\mathcal P$. Then the following statements on  the three-dimensional system hold: 
\begin{enumerate}
\item Solutions with initial data such that $x^0\neq x^+(\phi^0)$ have the property that 
$|\mathcal R_1(\phi(t), x(t))|$ decreases with respect to $t$, for sufficiently large $t>0$. In particular, $|\mathcal R_1|$ is strictly decreasing for $\phi^0\in(0, \phi_s^1)\cup (\phi_s^2, \phi^*)$, for all $t>0$. 
\item Let $\epsilon>0$ be sufficiently small. Then  $\mathcal M$ is asymptotically stable, that is, there exists $(\hat\phi(t), \hat x(t), \hat z(t))\in \mathcal M$ such that, for sufficiently large $t$, solutions corresponding to initial data  
$(\phi^0, x^0, z^0)\in \mathcal I$  satisfy
\begin{equation}
\phi(t, \epsilon)= \hat\phi(t)+ O(e^{-\frac{t}{\epsilon}}),\,  \, x(t, \epsilon)= \hat x(t)+ O(e^{-\frac{t}{\epsilon}}),\,\,  z(t, \epsilon)= \hat z(t)+ O(e^{-\frac{t}{\epsilon}}).  \label{estimate}
\end{equation}
\end{enumerate}
\end{proposition}
\begin{proof}Part 1 follows directly from the first equation in (\ref{R2}), the properties of the functions $H(\phi)$ and $R(\lambda)$ and (\ref{dphi}).  
The functions $(\hat\phi(t), \hat x(t), \hat z(t))$ are obtained as weak solutions of the two dimensional system. Their construction applies Definition \ref{weak-solution}, and estimate (\ref{estimate}) follows from Theorem \ref{convergence} on multiscale analysis of the system. 
\end{proof}

We  characterize weak solutions of the two dimensional system. These correspond to hysteresis loops  on the left graph of figure \ref{Hfig}.  One main ingredient in the construction is the theorem on continuation and finite time  blow-up of solutions of ordinary differential equations together with the existence of unique, unstable, equilibrium point.   This, together with the Poincar{\'e}-Bendixon theorem for two dimensional systems, leads to the existence of  a  limit cycle for such a system. The latter is also the $\omega$-limit set of the positive semi-orbits of the three dimensional system. 
\begin{definition} \label{weak-solution}
A  weak solution  of the two dimensional system corresponding to initial data $(\phi^0, x^0, z^0)\in\mathcal M$ has the following properties:
\begin{enumerate}
\item There exists $0<\hat t$ such that $(\phi(t), x(t), z(t))$ is a classical solution for $t\in(0, \hat t)$.
\item $\phi(\hat t)=\phi_s^1$ (or $\phi_s^2$).
\item $\phi(t)$ is discontinuous at $\hat t$ experiencing a jump
$[\phi(\hat t)]=\phi_s^2-\phi_s^1.$ (alternatively, $[\phi(\hat t)]=\phi_s^1-\phi_s^2.$)
\item The time derivatives of $\phi$ experience finite time blow up, that is,  $\lim_{t\to \hat t}\frac{d\phi}{dt}(t)=+\infty$ (alternatively, $\lim_{t\to \hat t}\frac{d\phi}{dt}(t)=-\infty$).
\item $x(t)$ and $z(t)$ are continuous at $\hat t$ and their derivatives experience a jump discontinuity.
\item The solutions can be continued for $t>\hat t$ and are  bounded. 
\end{enumerate}
\end{definition}
\begin{remark} We point out that  $\mathcal M^{-}$ is also an invariant manifold of the two-dimensional system.  In particular, it contains the unique stationary point, with eigenvalues forming a complex conjugate pair, with positive real part.  The limit cycle of the 2 dimensional system found next is also the $\omega$-limit set of solutions with initial data in $\mathcal M^-$.

We point out that values of $\phi\in(\phi^1_s, \phi_s^2)$ are excluded from the range of  solutions of the two-dimensional system. Indeed, in the next section,  we construct weak solutions of the system such that $\phi$  experiences jump discontinuities,  so as to avoid this interval. However, this interval is accessible to solutions of the full system, and, in particular, it contains the unique unstable equilibrium point. We will also show that this interval is covered in the fast time scale. 
\end{remark}
The following result follows from the uniqueness of classical solutions of the initial value problem of the system (\ref{dphi})-(\ref{Rlambda}). 

We denote the weak solution as  $\boldsymbol \psi_{\mathcal M, t}= (\phi_{\mathcal M}, x_{\mathcal M}(\phi), z_{\mathcal M}) $, where $\boldsymbol \psi_{\mathcal M, t} $ is the flow. 
Let $\boldsymbol\varphi_t$ denote the flow of the three dimensional system with initial data in $\mathcal I$.
\begin{theorem}  For each set of initial data $(\phi^0, x^0, z^0)\in\mathcal M$, there is a a unique weak solution $(\phi_{\mathcal M}(t), x_{\mathcal M}(t), z_{\mathcal M}(t)) $ of the two dimensional system, that exists for all $t>0$. Moreover,  the $\omega$-limit set  $\omega(\pi^+)$ of a semi-orbit $\pi^+$ of the three dimensional  system with initial data in $\mathcal I$ satisfies
$
\omega(\pi^+)=\omega(\pi^+_{\mathcal M}),
$  where $\omega(\pi^+_{\mathcal M})$ is the $\omega$-limit set  of the  trajectories of the two dimensional system. 
\end{theorem}
\begin{proof}
Without loss of generality let us assume that   $(\phi(0), x(0), z(0)) \in \mathcal M_2$ and with 
$ \frac{1}{2}{\lambda(x(0), \phi(0))}z(0)-x(0)<0.$ Then, there is $\hat t>0$ such that $(0, \hat t)$ gives the maximal interval of existence of the classical solution of the 2-dimensional system. Two different situations may occur, according to  the choice of initial data:
\begin{enumerate} 
\item $\hat t=\infty$ in which case the solution is bounded and such that $x(t)\in \mathcal M_2$ for all time, or
\item $0<\hat t<\infty$, in which case $x(\hat t)= x_s^2$, and $\phi(\hat t)=\phi_s^2$.
\end{enumerate}
In case (2), we further distinguish two cases, also according to initial data:

\noindent
{\bf {\small{2a}}.\,} $ \frac{1}{2}{\lambda(x(\hat t), \phi(\hat t))}z(t)-x(\hat t)=0$, in which case $\frac{dx}{dt}(\hat t)=0$. The solution can then be continued as in case 1, that is, as classical solution in $\mathcal M_2$, and, so it never reaches $\mathcal M_1$.

\noindent
{\bf {\small{2b}}.\,}
 $ \frac{1}{2}{\lambda(x(\hat t), \phi(\hat t))}z(t)-x(\hat t)<0$, and so  $\frac{dx}{dt}(\hat t)<0$. Since $\lim_{t\to\hat t}\frac{dx}{d\phi}(t)=0$, then $\lim_{t\to\hat t}\frac{d\phi}{dt}(t)= -\infty$. 
We set $\phi(\hat t^+)=\phi_s^1$, so  the jump condition  $[\phi]= \phi_s^1-\phi_s^2$ is satisfied,  $x(\hat t^-)= x(\hat t^+)=x_s^2$  and   $z(\hat t^-)= z(\hat t^+)$. Note that at $ t=\hat t^+$,  $(\phi, x, z)\in \mathcal M_1$. 

To continue the solution for $t>\hat t$,  we solve the initial value problem  with initial data $(\phi(\hat t^+), x(\hat t^+), z(\hat t^+))$. It is easy to see that there is a time $\tilde t>\hat t $ such that 
$\lim_{t\to \tilde t^-} \phi(t) = \phi_s^1$ and   $ \lim_{t\to \tilde t^-} x(t)=x_s^1$.  In fact, it follows from the fact that $\phi$ and $x$ remain bounded away from 0, so that the values $\phi_s^1$ and $x_s^1$, respectively, can be reached, allowing the solution to be continued in $\mathcal M_2$.  The fact that $\phi<\phi^*$ and  that there are no equilibrium points in $\mathcal M_2$ guarantees the existence of a point of return on the trajectory $(\phi(t), x(t), z(t))$, and so the process can be continued.

Since the orbits $\pi^+$ and $\pi^+_{\mathcal M}$ are both bounded, the scaling with respect to $\mathcal A_4$ holds and so does  the estimate
\begin{equation} |\phi(t)-\phi_{\mathcal M}(t)|=O(e^{-\frac{t}{\epsilon}}), \label{estimate1}\end{equation}
together with the analogous estimates for $x(t), z(t)$ which completes the proof of the theorem.
\end{proof} 

Since $\mathcal M$ is an invariant set  that does not contain equilibrium points, the existence of a limit cycle for the two dimensional system follows from  the Poincar{\'e}-Bendixon theorem stated next. Its corollary gives the stability of the limit cycle with respect to, both, the two-dimensional and the three-dimensional dynamics. 
\begin{theorem}[Poincar{\'e}-Bendixon\cite{Hale1980}] If $\pi^+$ is a bounded semiorbit and $\omega(\pi^+)$ does not contain any critical point, then either
 $ \pi^+=\omega(\pi^+),$ or
 $ \omega(\pi^+)={\bar\pi^+}/\pi^+$.
In either case, the $\omega-limit$ set is a periodic orbit. Moreover,
 for a periodic orbit $\pi^0$ to be asymptotically stable it is necessary and sufficient that there is a neighborhood $G$ of $\pi^0$ such that $\omega(\pi(p))=\pi^0$ for any $p\in G$. 
 \end{theorem}
\begin{figure}
\begin{minipage}[t]{\textwidth}
\centerline{
\includegraphics[scale=.40]{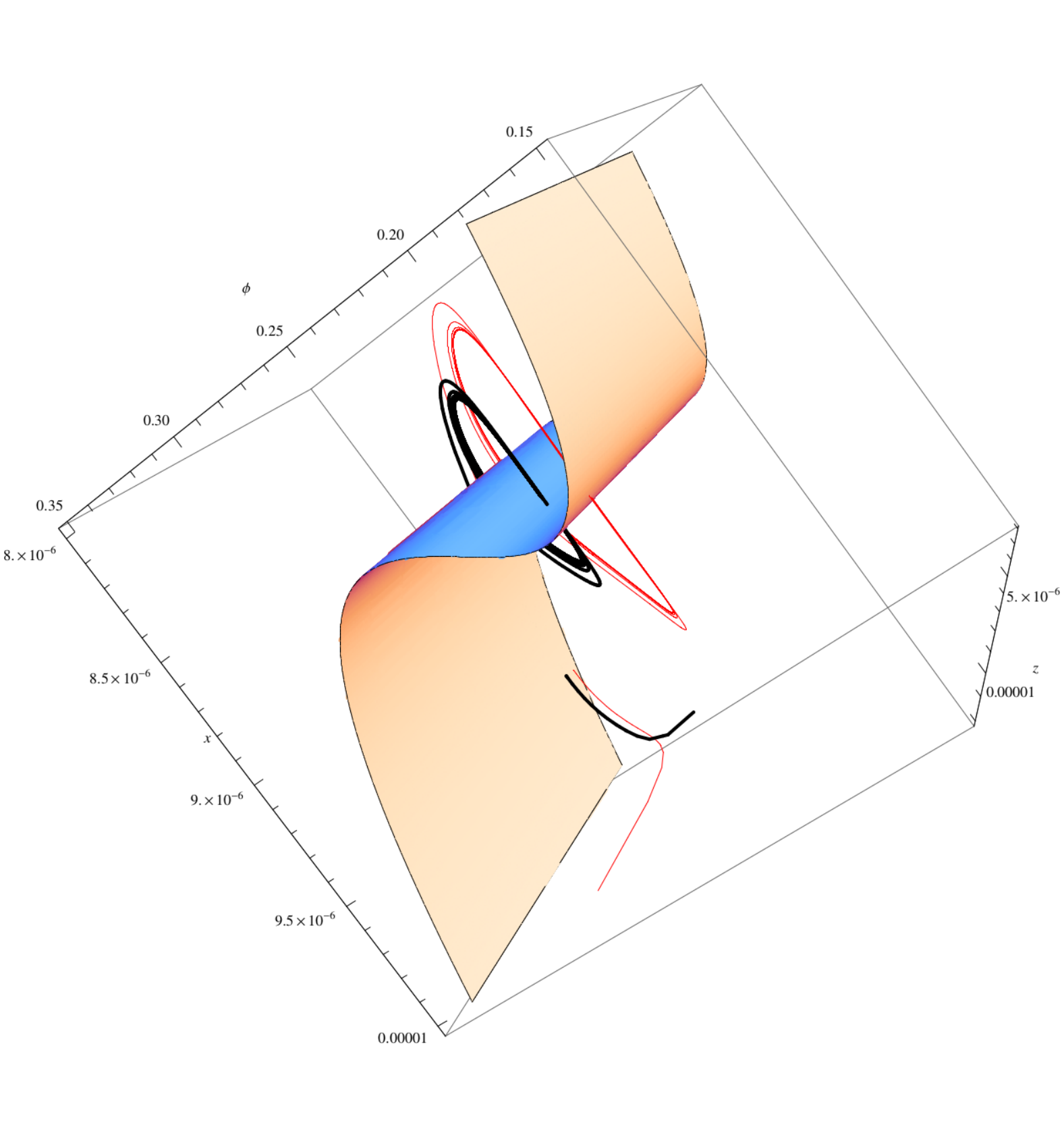}
\hspace{.60cm}
\includegraphics[scale=.40]{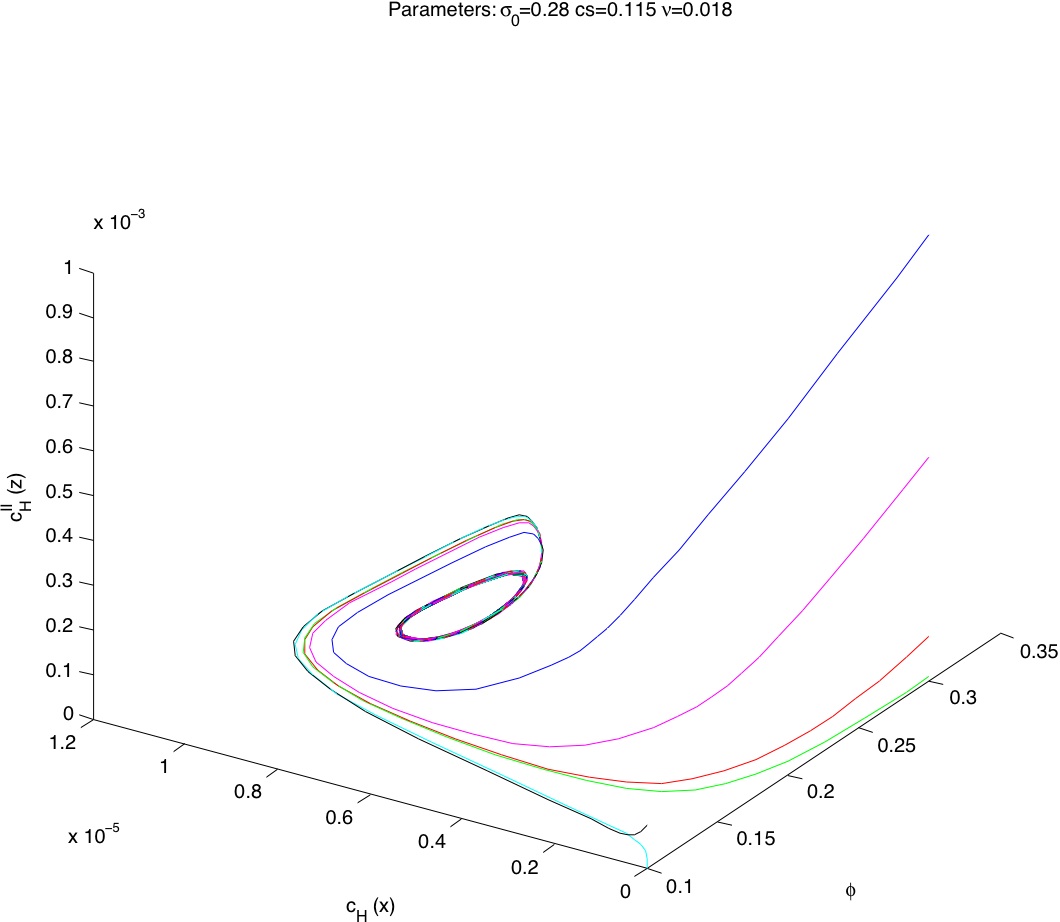}
}
\end{minipage}
\caption{The left figure shows orbits of the three dimensional system  with the manifold $\mathcal M$.  The figure on the right shows plots of orbits of the three dimensional system approaching a plane closed curve.
It is the $\omega$-limit set of  both systems and the limit cycle of the two dimensional system in the manifold $\mathcal M$. }\label{3d-curves}
\end{figure}
\begin{remark} 
The orbits of the two-dimensional system analyzed in this section, and in particular the limit cycle, correspond to the projection to the $x-z$ plane of the three-dimensional orbits shown in
figure \ref{3d-curves}(b).   
Moreover, these two-dimensional orbits correspond to the leading terms in expansions (\ref{ilx})-(\ref{ilz}). 
\end{remark}
 Next theorem shows that the limit cycle of the two dimensional system is also the limit cycle of the three dimensional one. 
 \begin{theorem}   Suppose that the parameters of the system belong to the class $\mathcal P$. Then the  two-dimensional system  has an asymptotically stable  limit cycle $\pi^+$ in $\mathcal M$.  
 Moreover $\pi^+$ is also the $\omega$-limit set of   positive semi-orbits of the three-dimensional system with initial data  in $\mathcal I$.   
  \end{theorem}
 \begin{proof} Let us consider solutions with initial data in $\mathcal M$. For these initial data, we previously showed   that the two-dimensional system admits bounded weak solutions that exist for all time. Moreover, the orbits of these solutions do not contain any equilibrium point. So, the existence of a limit cycle $ \pi^+$  in $\mathcal M$ follows directly from the Poincar{\`e}-Bendixon theorem. The asymptotic stability of $\pi^+$ follows from the boundedness of solutions and the absence of a stationary state in $\mathcal M$. 
 To prove the last statement, let us consider solutions with initial data $(\phi^0, x^0, z^0)\in \mathcal I$ and such that $x^0\neq x^+(\phi^0)$. By an estimate as (\ref{estimate1}), we can assert that sufficiently large $t$, 
$
  |x(t)-x_M(t)|= O(e^{-Mt}), $
 with $M>0$ as in (\ref{M1}), and $x_M(t):=x_M(\phi^0, x^+(\phi^0), z^0)(t)$. This indicates that the solution of the three dimensional system approaches the two-dimensional manifold, 
 for sufficiently large $t$. Since the only equilibrium point in the three-dimensional space is unstable, then the three-dimensional solution also approaches the limit cycle (figure \ref{3d-curves}(b)). 
  \end{proof} 
  

\begin{remark}
  We observe that the solutions of the plane system, and in particular the limit cycle, are discontinuous with respect to $\phi$. In particular, $\phi$ is discontinuous at $t=\hat t$, with $[\phi(\hat t)]\neq 0$.
  The regularization of the solutions is done by {\it connecting} the discontinuity values of $\phi$ by a function $\bar\phi(\frac{t}{\mathcal A_4})$  that evolves according to the {\it  fast } dynamics presented  in the next section. 
   For sufficiently small $\mathcal A_4$, the solutions of the three-dimensional system remain in $\mathcal M$ for most of the time,  emerging  to the third dimension in order to connect
   the separate branches  $\mathcal M_1$ and $\mathcal M_2$, in the fast time scale. This is developed in section \ref{multiscale}. 
  \end{remark}
\section {Competitive systems and three-dimensional limit cycle}{\label{3D-limit-cycle}}
  It is well known that for three dimensional systems,  the compactness of a steady state free $\omega-$limit set  of an orbit is not sufficient to prove the existence of a periodic orbit. 
  That is, the Poincar{\'e}-Bendixon theorem in its original form does not apply. However, a three-dimensional generalization is available for {\it competitive }
  (and {\it cooperative}) systems (\cite{Hirsh1},  \cite{Hirsh2}, \cite{Hirsh4}, \cite{Smith-survey}).
 This  concept is framed in terms of monotonicity properties   of the vector field of the system with respect to a convex subspace of ${\mathbb R}^3$. 
  The theorem is due to Hirsch and it is stated as follows:
 \begin{theorem} \label{PB-theorem-3d}
  A compact limit set of a competitive or cooperative  system in $ {\mathbf R^3}$  that contains no equilibrium points is a periodic orbit.
 \end{theorem}

 A dynamical system  is called {\it competitive} in the positive cone  ${\mathbb{R}^+}^{n} $
if
$ \frac{\partial f_i}{\partial x_j}(\bx)\leq 0, \  i\neq j, \, \bx\in \mathcal I.$ 
This property holds in a general cone $\mathcal K$ (intersection of half-spaces), by requiring the following two properties: 
 \begin{definition}\label{sign-symmetry} A $n\times n$ Jacobian matrix $(\nabla\boldsymbol f)$ is {\it{sign-stable}} in $\mathcal I$ if for each $i\neq j$, either $\partderiv{f_i}{x_j}\geq 0 $ or $\partderiv{f_i}{x_j}\leq 0 $, for all $\bx\in\mathcal I$.
It is {\it{sign-symmetric}}
 if $\partderiv{f_i}{x_j}(\bx)\partderiv{f_j}{x_i}(\by)\geq 0 $ for all $i\neq j$ and for all $\bx, \by \in\mathcal I$. 
 \end{definition}
 \begin{proposition} The three dimensional governing system (\ref{dphi})-(\ref{electroneutrality-soln}) is competitive with respect to the cone
$
 \mathcal K=\{(\phi, x, z):  \phi<0, \, \, x>0, \,\, z<0\}.$
  \end{proposition}
  
  The proof involves identifying $\mathcal K$ and verifying definition (\ref{sign-symmetry}), which follows from these  two lemmas.
\begin{lemma}  \label{inequalities1-2} The inequalities
$
\phi(1-\phi)\frac{\partial\lambda}{\partial \phi}=\lambda \frac{fp}{\sqrt{1+p^2f^2}}\geq \lambda  \,\, \textrm{and}\,\, \frac{\partial\lambda}{\partial x} <0
$
hold on 
 trajectories $\boldsymbol \varphi_t$ corresponding to initial data in $\mathcal I$. 
\end{lemma} 
\begin{proof} For $\lambda, p $ and $f$ as in (\ref{electroneutrality-soln}) and (\ref{gamma-gamma0}),   simple calculations show that 
\begin{eqnarray}
&& \frac{2}{\gamma}(1-\phi)^2\frac{\partial\lambda}{\partial \phi}=\frac{f}{\sqrt{1+p^2f^2}}=\frac{f\lambda}{\sqrt{1+p^2f^2}}= \frac{2}{\gamma}\frac{(1-\phi)}{\phi}\frac{f p\lambda}{\sqrt{1+p^2f^2}},\nonumber\\
&&\frac{\partial\lambda}{\partial x}=\frac{p f'}{\sqrt{1+p^2f^2}}\lambda= -\frac{p}{(1+x)^2}\frac{\lambda}{\sqrt{1+p^2f^2}}> -\frac{p}{(1+x)^2}\lambda,
 \end{eqnarray}
 from which the stated inequalities follow. 
\end{proof}
 \begin{lemma} The  Jacobian matrix $\mathcal J$ of the three-dimensional system is sign-stable and sign-symmetric. \end{lemma}
 \begin{proof}  
Let us calculate the Jacobian matrix $J:=\{a_{ij}\} $ at an arbitrary state $(\phi, x,z)$:
\begin{eqnarray}
&&a_{12}=  -(1-\phi)\phi^2R'(\lambda)\frac{\partial\lambda}{\partial x}, \quad
a_{21}= (1-\phi)(1+x)^2\frac{z}{2}\frac{\partial\lambda}{\partial\phi}-(1+x)^2(\frac{\lambda}{2}z-x),\\
&&a_{31}=  -\beta e^{-\beta\phi}-\frac{\cA_2}{\cA_1}\big((1-\phi)z\frac{\partial\lambda}{\partial\phi}-(\lambda z-x)\big), \quad a_{13}=0,\\
&& a_{23}= (1-\phi)(1+x)^2\frac{\lambda}{2}, \quad a_{32}= -\frac{\cA_2}{\cA_1}(1-\phi)(z\frac{\partial\lambda}{\partial x}-1).
\end{eqnarray}
Note that  the off-diagonal elements of the matrix $J$   have the signs:
\begin{eqnarray}
 &&a_{12}>0, \, a_{13}=0, \, a_{23}>0,  \, a_{32}>0,   \,  a_{21}>0,
\, a_{31}<0.\nonumber
\end{eqnarray}
from which, sign-symmetry and sign-stability immediately follow.
For this, let us write
\begin{equation}
a_{21}=(1+x)^2\big((1-\phi)\frac{\partial\lambda}{\partial \phi}\frac{z}{2}-\frac{\lambda}{2}z+x\big)> \frac{x}{(1+x)^2}>0. 
\end{equation}
Note that $a_{23}>0$   and 
$ a_{31} <-\beta e^{-\beta\phi}-\frac{\cA_2}{\cA_1} x<0, $
where we have used the first inequality in lemma \ref{inequalities1-2}.  Finally, to check the sign of $a_{12}$,  we differentiate $R(\lambda) $ in (\ref{Rlambda}) to find
$ R'(\lambda)= (\sqrt{\lambda}-\frac{1}{\sqrt{\lambda}})(\lambda^{-\frac{1}{2}}+ \lambda^{-\frac{3}{2}})$, from which  the sign of $a_{12}$ follows, taking into account the second inequality in lemma \ref{inequalities1-2}. The identification of the corresponding cone $\mathcal K$ is done in the Supplementary Material section. 
\end{proof}


  One of the practical difficulties in the application of the Poincar{\`e}-Bendixon theorem to three-dimensional systems as compared to the two-dimensional counterpart is the verification that 
 the $\omega$-limit set does not contain equilibrium points.  The case that the system has  a single unstable equilibrium point is the simplest one  to treat.
%
  Existence of a three dimensional limit cycle, together with additional properties of competitive systems  relevant to the current analysis are summarized next.
 \begin{theorem}  Let $\mathcal P$, $\mathcal M$ and $\mathcal I$ be as in definition (\ref{parameter-space}), (\ref{Mminus}) and (\ref{I}), respectively.  Then
 \begin{enumerate}
 \item The flow  on the $\omega$-limit set of orbits of the three-dimensional system in $\mathcal I$ is topologically equivalent to the flow on the $\omega$-limit set of orbits of the two-dimensional system in $\mathcal M$.
 \item 
  Let $\bp=(\Phi, X, Z)$ be the unique equilibrium point of the system. If $\bp\neq \mathbf q \in \mathcal I$, then 
$\bp\notin\omega(\mathbf q). $ 
 \item The  three-dimensional system has a limit cycle. 
 \end{enumerate}
 \end{theorem}
 \begin{proof}
 The first item follows from the fact that the two-dimensional system is Lipschitz, the competitiveness of the three dimensional system,  and from the  compactness of the $\omega$-limit set of both systems. 
    The second statement follows from competitiveness and the fact that the equilibrium point is unique and hyperbolic. 
In fact, it is a consequence of the theorem that states that {\sl a compact limit set of a competitive or cooperative system cannot contain two points related by $<<$}  (\cite{Smith-survey},  Theorem 3.2). Finally,  property (3) is a consequence of (2),  the compactness of the  $\omega$-limit set and Theorem \ref{PB-theorem-3d}. 
 \end{proof}
 
\begin{remark} In the case that the unique equilibrium point $\bp$  is hyperbolic, to prove that $\bp\notin\omega(\mathbf q)$, $\mathbf q\in \mathcal D$, it is necessary to assume that the system is competitive. This is the case encountered in some applications such as in virus dynamics \cite{virus}.
\end{remark} 
\section{Multiscale Analysis}{\label{multiscale}}
The analysis carried out in   the previous sections is mostly  based on the time scale separation between the mechanical and chemical evolution components of the system. In turn, this is based on   the relative sizes of the dimensionless parameters $\mathcal A_i$. The goal of this section is to find estimates for the solutions of the three dimensional system with respect to the fast time scale.  Let us consider the time scales $\bar t$ and 
$\tau$ as in (\ref{time-scales}), representing the  {\it fast} and {\it slow}  times, respectively.  We consider the system of equations (\ref{phi333})-(\ref{z333}), 
  together with (\ref{electroneutrality-soln})-(\ref{R3}).
We propose the following solution ansatz
\begin{eqnarray}
&&\phi(\bar t,\tau, \eps)=\Phi(\tau,\eps) + \tphi(\bar t,\eps)=\big(\sum_{j=0}^N\Phi_j(\tau)+\mathcal E_0^1\big)+ \,\big(\sum_{j=0}^N\tphi_j(\bt)+\mathcal E_I^1\big), \label{ilphi}  \\
&&x(\bar t,\tau,\epsilon)= X(\tau,\epsilon) + \tilde x(\bar t,\epsilon)= \big(\sum_{j=0}^NX_j(\tau)+\mathcal E_0^2\big)+ \,\big(\sum_{j=0}^N\tx_j(\bt)+\mathcal E_I^2\big) \label{ilx}\\
&&z(\bar t,\tau,\eps) =Z(\tau,\eps) +\tz(\bar t,\eps) + \big(\sum_{j=0}^NZ_j(\tau)+\mathcal E_0^3\big)+ \,\big(\sum_{j=0}^N\bar z_j(\bt)+\mathcal E_I^3\big) , \label{ilz}
\end{eqnarray}
with 
 $\eps=\mathcal A_4, \,\, \bar t=\frac{\tau}{\eps},$
and $\mathcal E_0^j, \mathcal E_I^j$, $j=1 ,2 ,3,$ denote error terms.
We now consider $\frac{\cA_4}{\cA_1}$ fixed and $0<\eps<<1$.  Equations for the terms in (\ref{ilphi})-(\ref{ilz}) are obtained as follows:

\noindent {\bf {\small 1}.\,}
Holding $\tau$ fixed, and letting $\eps\to 0$ yields equations for $\Phi, X, Z$. These  are studied in section (\ref{inertial-manifold}). 

\noindent{\bf {\small 2.}\,}
 Holding $\bar t$ fixed, and  letting $\eps\to 0$   yields equations for  the initial layer term $\tphi, \tx $ and $\tz$. Moreover, we seek solutions with the following asymptotic property
$\{\tphi(t,\eps),  \, \tx(t,\eps), \, \tz(t,\eps)\}=O(\exp(-ct)) \,  \textrm{as} \, t\to \infty,$
where $c>0$ is a material dependent parameter. Note that since $X$ and $Z$ satisfy initial conditions at $\tau=0$, then 
$\tx=0=\tz,$ 
and it is necessary to calculate only the initial layer for $\tilde\phi$.
\subsection{Solutions in the  fast time scale $\bar t$}
 To obtain equations for the initial layer terms, we substitute expressions (\ref{ilphi})-(\ref{ilz}) into the governing equations,  fix $\bt>0$ and take the limit $\eps\to 0$. In particular, this yields the limit $\tau=0$. 
Moreover, taking into account that $\cR_1(\Phi, \Lambda)=-(1-\Phi)\Phi^2\big(R(\Phi)-H(\Lambda)\big)=0$,  and linearizing about $(\Phi, \Lambda)$ gives the following 
 equation for  $\tphi(\bar t,\tau,\eps):$
  \begin{equation}\label{M1}
\frac{d\tphi}{d\bar t}= -M(\Phi,\Lambda)\tphi(\bar t,\eps)+ o(\tphi), 
\quad
 M(\Phi, \Lambda)
 := \frac{\partial \cR_1}{\partial \phi}(\Phi, \Lambda)+ \frac{\partial \cR_1}{\partial \lambda}(\Phi, \Lambda)\frac{\partial\lambda}{\partial\phi}. 
\end{equation}
We point out that  $(\Phi, \Lambda, Z)$ solve the two-dimensional system.
To calculate the right hand side of (\ref{M1}), we  take  derivatives in equation (\ref{lambda-p-q}),
$\frac{\partial\lambda}{\partial \phi}(\Phi, X)= p'(\Phi) f(X) + q'(\Phi), \,\, R'(\Lambda)=\gamma_0(1-\frac{1}{\Lambda^2}).$
 So,
\begin{eqnarray}
M(\Phi, \Lambda)&&= -(1-\Phi)\Phi^2\big(H'(\Phi)-R'(\Lambda)(p'f(X) +q'(\Phi))\big)\nonumber\\ 
&&= -\frac{1}{2}(1-\Phi)\Phi^2\big(H'(\Phi)(1+\frac{1}{\Lambda^2}) -\frac{\gamma\gamma_0}{2(1-\Phi^2)}f(X)(1-\frac{1}{\Lambda^2})\big). \label{M2}
\end{eqnarray}
We point out that according to the earlier observation that $(\Phi, \Lambda)$ are evaluated at $\tau=0$, it follows that the rate of decay of the fast component of the solution depends on  the projection of the initial data on the slow manifold. 
\begin{lemma} Consider initial data $(\phi^0, \lambda^0, z^0)\in\mathcal I$ and let  $M^0:=M(\phi^0, \lambda^+(\phi^0), z^0)$.
If $M_0>0$, then the solution of the linear equation (\ref{M1}) satisfies $\tilde\phi(\bt)=O(e^{-M^0 \bt}), $ for $\bt>0$ large. Otherwise,  
 $\tilde\phi(\bt)=O(e^{|M_0| \bt}), $ for $\bt>0$ large. 
\end{lemma}

This lemma motivates the definition of {\it slow} and {\it fast}  manifolds  of the (three-dimensional) system. Let 
\begin{eqnarray}
&& \mathcal S= \{(\phi, \lambda, z): \lambda=\lambda^+(\phi),  \, M>0\}, \quad \mathcal F=  \mathcal I/\mathcal S, 
\end{eqnarray}
with $\lambda^+$ as in (\ref{lambdapm}). Note that the decomposition $\mathcal S\bigoplus\mathcal F$ corresponds to that in (\ref{ilphi})-(\ref{ilz}).
\begin{remark}
Note that a sufficient condition for $M^0>0$ is that $H'(\phi^0)\leq0$, that is $\phi^0\in \mathcal N^{-}$. So, the first inequality may still hold in the case that $H'(\phi^0)>0$, that is, for $\phi_0>\phi_{\textrm{\tiny{min}}}$ or $\phi_0<\phi_{\textrm{\tiny{max}}}$ in figure \ref{Hfig}. This property is referred to as the reduced system having a {\it can{\`a}rd} structure \cite{Bold-canards2003, Showalter-canards1991}.
The case $M=0$ cannot be characterized in terms of linear stability. 
\end{remark}
\begin{theorem}\label{convergence}
Let $(\phi^0, x^0, z^0)\in\mathcal I$ and suppose that the parameters of the system belong to the class $\mathcal P$. 
Then there exists $\epsilon_0>0$ such that for $\epsilon\in(0,\epsilon_0) $ the governing system has a unique $C^1$-solution $(\phi(t; \epsilon), x(t; \epsilon), z(t; \epsilon))\in\mathcal I$, for all $t>0$. Moreover $(\phi(\cdot), \lambda(\cdot), z(\cdot))\in  \mathcal S\bigoplus\mathcal F  $  and  it admits the asymptotic expansions (\ref{ilphi}), (\ref{ilx}) and (\ref{ilz}), with the property that 
$\mathcal E_0^k= O(\eps^N), \,\, \mathcal E_I^k=O(e^{-M_0\bt}), \,\, k=1,2,3,$
for sufficiently large $\bt$.
\end{theorem}
The proof  follows from linearization of the three-dimensional system with respect to solutions of the two dimensional one, subsequently transforming it into a system of integral equations, and applying Schauder's fixed point theorem to it. 
\begin{remark}
We consider the  model  obtained  by further reducing the time scale, that is, setting the limiting problem in the {\it slow } scale as
$0= \cR_1(\phi, \lambda), \,\,
\frac{d\bar x}{d\tau}=\cR_2(\phi,\bar x,\lambda), \,\,
0=\cR_3(\phi,\bar x, \bar z,\lambda),$ 
together with equations (\ref{electroneutrality-soln}).
This is consistent with taking $\frac{\cA_4}{\cA_1}=O(\eps^{-3})$. 
 This model yields that proposed by Siegel and Li \cite{li-siegel99}, that sets a rely equation for the {\it product} of the reaction, in this case, the hydrogen ion concentration in the membrane. 
\end{remark}
\begin{figure}\label{oscillations-lingxing-april3}
\centerline{
\includegraphics[scale=.4]{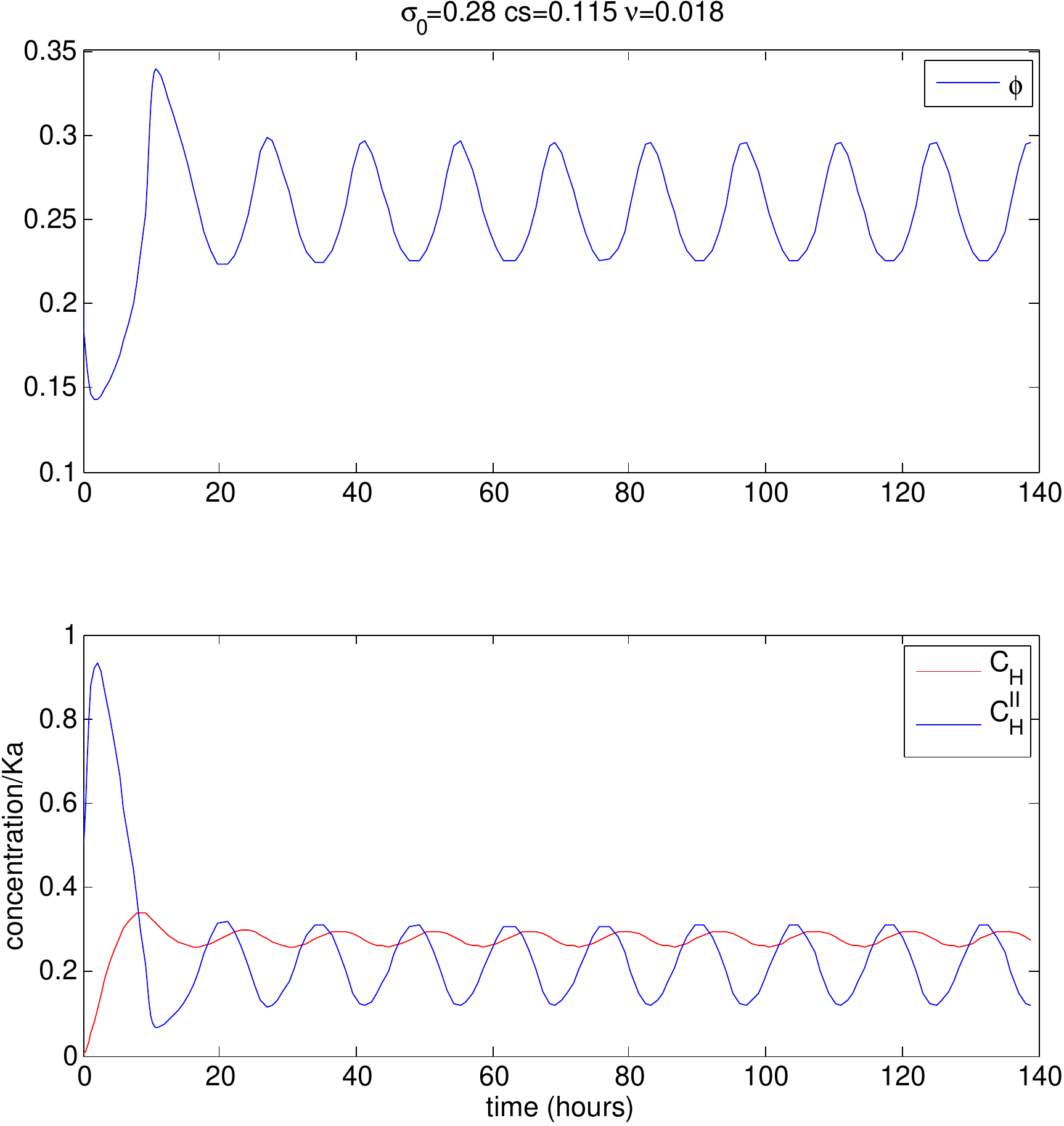}} 
\caption{The top plot represents the swelling dynamics of the membrane in terms of the polymer volume fraction. The graphs of the bottom plot refer to the evolution of ${\textrm H}^+$ in
the membrane (red) and in cell II (blue), respectively. The oscillatory behavior is compatible with the GnRH pulse release.}
\end{figure}
\section {Concluding Remarks}

We have analyzed a {\it lumped} model for a chemomechanical oscillator suitable for rhythmic drug delivery.  The model consists of a system of ordinary differential equations for the chemo-mechanical fields. For this system, we showed existence of periodic solutions, which correspond  to experimentally and numerically observed oscillations. The tools of the analysis involve multiscale  and dynamical systems methods,  including the theory of competitive dynamical systems \cite{Hirsh1}. 

The  membrane model, which ignores gradients of solute concentrations and swelling within the membrane, can be replaced by a distributed, PDE based system which, in addition to more accurately portraying the physical situation, can include self consistent, natural boundary conditions at the interfaces between the membrane and the two chambers. The results presented here will not be altered qualitatively, though there will be quantitative differences.


Overall, the model presented here dealt with the fundamental mechanisms underlying oscillatory behaviour of a table-top experimental device. It must be regarded as a first step, since complications associated with the buildup of gluconate ion in the system, which buffers and affects the dynamics of pH oscillations, have not been included.  Also the effects of endogenous phosphate and bicarbonate buffering species would need to be included in a more comprehensive model, which would be of higher dimensionality, even in the lumped framework.  These buffering effects are currently the main hurdle to develop an in-vivo device. 
%
\bibliography{gel}
\bibstyle{siam} 
\end{document}